\documentclass[11pt,reqno]{amsart}

\newcommand\version{June 18, 2019}

\usepackage{amsmath,amsfonts,amsthm,amssymb,amsxtra}

\newtheorem{theorem}{Theorem}[section]
\newtheorem{proposition}[theorem]{Proposition}
\newtheorem{lemma}[theorem]{Lemma}
\newtheorem{corollary}[theorem]{Corollary}

\theoremstyle{definition}

\newtheorem{assumption}[theorem]{Assumption}

\theoremstyle{remark}

\newcommand{\C}{\mathbb{C}}

\renewcommand{\epsilon}{\varepsilon}

\renewcommand{\phi}{\varphi}
\newcommand{\R}{\mathbb{R}}

\DeclareMathOperator{\im}{Im}

\DeclareMathOperator{\re}{Re}

\begin{document}

\title[Landau--Pekar equations --- \version]{A non-linear adiabatic theorem for the one-dimensional Landau--Pekar equations}

\author{Rupert L. Frank}
\address[Rupert L. Frank]{Mathematisches Institut, Ludwig-Maximilans Univers\"at Mu\"unchen, Theresienstr. 39, 80333 M\"unchen, Germany, and Munich Center for Quantum Science and Technology (MCQST), Schellingstr. 4, 80799 M\"unchen, Germany, and Mathematics 253-37, Caltech, Pasa\-de\-na, CA 91125, USA}
\email{rlfrank@caltech.edu}

\author{Zhou Gang}
\address[Zhou Gang]{Department of Mathematics, Binghamton University, Binghampton, NY 13902-6000}
\email{gzhou@math.binghamton.edu}

\begin{abstract}
We discuss a one-dimensional version of the Landau--Pekar equations, which are a system of coupled differential equations with two different time scales. We derive an approximation on the slow time scale in the spirit of a non-linear adiabatic theorem. Dispersive estimates for solutions of the Schr\"odinger equation with time-dependent potential are a key technical ingredient in our proof.
\end{abstract}

\maketitle

\renewcommand{\thefootnote}{${}$} \footnotetext{\copyright\, 2019 by the authors. This paper may be reproduced, in its entirety, for non-commercial purposes.\\
The first author would like to thank Benjamin Schlein and Robert Seiringer for interesting discussions. Partial support through US National Science Foundation grant DMS-1363432 and through German Research Foundation grant EXC-2111 390814868 (R.L.F.) is acknowledged.}


\section{Introduction and main result}

A polaron is a physical model for a particle accompanied by its polarization field. We treat a one-dimensional, classical version of this model, where the electron is described by a complex-valued wave function $\psi\in L^2(\R)$ and the polarization field by a real-valued function $\phi\in L^2(\R)$. The strength of the coupling between the particle and the field is described by a constant $\sqrt\alpha = \epsilon^{-1/4}$ which is assumed to be large. While the original polaron model is three-dimensional, its one-dimensional version, which we discuss here, has been introduced in the physics literature both as a toy problem \cite{Gr} and as a limiting model for the three-dimensional model in a strong magnetic field, see \cite{KoLeSm,FrGe} and references therein.

Landau and Pekar \cite{LaPe} derived phenomenologically equations of motion for the polaron, whose one-dimen\-sio\-nal analogues read
\begin{align}\label{eq:lp}
\begin{cases}
\epsilon i \partial_t\psi& =-\partial_x^2\psi+\phi \psi \,,
\\
-\partial_t^2\phi & =\phi+\tfrac12 |\psi|^2 \,.
\end{cases}
\end{align}
Note that the typical time scale of the electron is of order $\epsilon$, whereas the time scale of the field is of order $1$.

Equations \eqref{eq:lp} are supplemented by initial conditions
\begin{equation}
\label{eq:lpinitial}
\psi|_{t=0} = \psi_0 \,,\qquad \phi|_{t=0} = \phi_0 \,,\qquad \partial_t\phi|_{t=0} = \dot\phi_0 \,,
\end{equation}
which we assume to be independent of $\epsilon$. By a standard argument (as, for instance, in \cite[Lemma 2.1]{FrGa}), using conservation of mass and energy, one can show that \eqref{eq:lp}, \eqref{eq:lpinitial} has global solutions for $\psi_0\in H^1(\R)$ and $\phi_0,\dot\phi_0\in L^2(\R)$.

Our goal is to approximate the dynamics on time scales of order one for a certain class of physically relevant initial conditions. Namely, under some assumptions, we will prove that if the initial wave function $\psi_0$ is the ground state of the Schr\"odinger operator $-\partial_x^2+\phi_0$ with the initial field as a potential, then up to times $t$ of order $1$, the wave function $\psi_t$ at time $t$ is close to the ground state of the Schr\"odinger operator $-\partial_x^2 +\phi_t$ with the field at time $t$ as potential. More precisely, we will construct $\epsilon$-independent limiting dynamics $(V,Q)$ where $Q$ is an exact ground state of $-\partial_x^2+V$ such that up to times of order $1$ the solution $(\psi,\phi)$ is in a quantitative sense well approximated by $(Q,V)$ after multiplying $\psi$ by an explicit phase. As we will discuss below in some more detail, this result is in the spirit of a non-linear adiabatic theorem.

Let us state our main result in detail. We will work under the following assumption on the initial data.

\begin{assumption}\label{ass:main}
Let $\phi_0\in L^2(\R)\cap \langle x\rangle^{-2} L^1(\R)$ be real-valued and assume that the Schr\"odin\-ger operator $-\partial_x^2+\phi_0$ in $L^2(\R)$ has a unique negative eigenvalue and no zero-energy resonance. We denote the eigenvalue by $E_0$ and a corresponding real-valued eigenfunction (not necessarily normalized) by $\psi_0$. Moreover, let $\dot\phi_0\in L^2(\R)\cap \langle x\rangle^{-2} L^1(\R)$ be real-valued.
\end{assumption}

We recall (see, e.g., \cite[Chapter 5]{Ya}) that the Schr\"odinger operator $-\partial_x^2+W$ is said to have a \emph{zero-energy resonance} if there is a non-trivial, \emph{bounded} function $u$ on $\R$ such that $(-\partial_x^2 + W)u=0$. We recall that if $W\in \langle x \rangle^{-1} L^1(\R)$, then any solution $u$ of the latter equation satisfies $u(x) \sim b_\pm x$ as $x\to\pm\infty$, and so having a zero-energy resonance means that there is a non-trivial solution with $b_+=b_-=0$. Generically, $-\partial_x^2+W$ has \emph{no} zero-energy resonance.

The assumption that $\psi_0$ is real-valued is not restrictive since, because of the simplicity of $E_0$, any corresponding eigenfunction $\tilde\psi_0$ is of the form $e^{i\theta}\psi_0$ for a real-valued $\psi_0$, and then the pair $(e^{i\theta}\psi,\phi)$ is a solution of \eqref{eq:lp} with the initial condition $(\tilde\psi_0,\phi_0,\dot\phi_0)$.

Using a fixed point argument (see Proposition \ref{reference}) one can show that there is a maximal interval $[0,T_*)$, $T_*\in(0,\infty)\cup\{\infty\}$, as well as unique functions $Q\in C^\infty([0,T_*),H^1(\R,\R))$, $V\in C^\infty([0,T_*),L^2(\R))$ and $E\in C^\infty([0,T_*),(-\infty,0))$ such that for all $t\in [0,T_*)$
\begin{equation}
\label{eq:eqref}
(-\partial_x^2 + V)Q = EQ \,,
\qquad
-\partial_t^2 V = V + \tfrac12 Q^2
\end{equation}
and
\begin{equation}
\label{eq:refmass}
\| Q \|_2 = \|\psi_0\|_2
\end{equation}
and
\begin{equation}
\label{eq:refinitial}
Q|_{t=0} = \psi_0 \,,
\qquad
V|_{t=0} = \phi_0 \,,
\qquad
\partial_t V|_{t=0} = \dot\phi_0 \,.
\end{equation}
We set
\begin{align}\label{eq:tstar}
T^* & :=\sup\left\{ T\in [0,T_*]:\ E_t \ \text{is the unique negative eigenvalue of}\ -\partial_x^2 +V_t \right. \notag \\
& \qquad \qquad \qquad\qquad \quad\ \left. \text{and there is no zero energy resonance for all}\ t\in(0,T) \right\}
\end{align}
and note that $T^*>0$.

The following is our main result.

\begin{theorem}\label{main}
Let $\psi_0,\phi_0$ and $\dot\phi_0$ be as in Assumption \ref{ass:main} and let $(V,Q)$ be the solution of \eqref{eq:eqref}, \eqref{eq:refmass} with initial data \eqref{eq:refinitial}. Then for every $T<T^*$ there is an $\epsilon_T>0$ such that for all $t\in[0,T]$ and $\epsilon\in (0,\epsilon_T]$, the solution $(\psi,\phi)$ of \eqref{eq:lp} with initial data \eqref{eq:lpinitial} satisfies
$$
\left\|\psi_t - e^{-i\epsilon^{-1}\int_0^t E_s \,ds} Q_t \right\|_2 \lesssim \epsilon \,,
\qquad
\| \phi_t - V_t \|_2 + \|\partial_t\phi_t - \partial_t V_t\|_2 \lesssim \epsilon^2 \,.
$$
More precisely, setting
$$
\alpha_t := \|\psi_0\|^{-2} \langle Q_t, e^{i\epsilon^{-1}\int_0^t E_s \,ds} \psi_t\rangle
\qquad\text{and}\qquad
R_t := e^{i\epsilon^{-1} \int_0^t E_s \,ds} \psi_t - \alpha_t Q_t \,,
$$
we have the decomposition
\begin{align}
\label{eq:decomp}
\psi_t = e^{-i\epsilon^{-1}\int_0^t E_s \,ds} \left( \alpha_t Q_t + R_t \right)
\qquad\text{with}\qquad
\langle Q_t,R_t\rangle = 0
\end{align}
and the bounds
$$
\left\| R_t \right\|_2 = \|\psi_0\|_2 \sqrt{1-|\alpha_t|^2} \lesssim \epsilon \,,
\qquad
\left| \partial_t \alpha_t \right|\lesssim \epsilon \,,
\qquad
\left| \partial_t (|\alpha_t|^2) \right| \lesssim
\begin{cases}
\epsilon & \! \text{if}\ t\leq \epsilon \,,\\
\epsilon (\epsilon/t)^{3/2} & \!\text{if}\ \epsilon\leq t\leq\epsilon^{1/3},\\
\epsilon^2 & \!\text{if}\ t\geq \epsilon^{1/3} \,.
\end{cases}
$$
\end{theorem}

We emphasize that $\alpha_t$ and $R_t$ depend on $\epsilon$, whereas $Q_t$ and $E_t$ do not.

The bound on the approximation of $\psi_t$ stated in the first part of the theorem follows from the bounds in the second part since
$$
\left\|\psi_t - e^{-i\epsilon^{-1}\int_0^t E_s \,ds} Q_t \right\|_2 = \left\| (\alpha_t-1) Q_t + R_t \right\|_2 = \sqrt{|\alpha_t-1|^2 \|\psi_0\|_2^2 + \|R_t\|_2^2} \lesssim \epsilon \,.
$$
We believe that the order $\epsilon$ is best possible, since in the proof of the theorem we will extract from $R_t$ a term which is a multiple of $\epsilon$ and show that the remainder is, at least in the norm of $\langle x\rangle L^\infty$ and for times $t\geq \epsilon^{1/3}$, bounded by $\epsilon^2$.

A result closely related to Theorem \ref{main} appears in \cite{LeRaScSe}. We'll discuss similarities and differences at the end of this introduction.

The statement of the theorem is reminiscent of the adiabatic theorem in quantum mechanics, which states that, under a gap condition, a system initially in an eigenstate remains close to its instanteneous eigenstate if the Hamiltonian changes slowly. Some recent works have explored to which extent this theorem remains valid for \emph{non-linear} Schr\"odinger equations. The paper \cite{Sp} studies the case of a weak non-linearity and modifies techniques from the proof of the (linear) adiabatic theorem. In contrast, we will follow the approach initiated in \cite{GaGr}, which exploits a completely different mechanism, namely that of \emph{dispersion}. It draws its inspiration from works on asymptotic stability of ground states of non-linear Schr\"odinger equations, a topic that was pioneered by Soffer and Weinstein \cite{SoWe1,SoWe2} and Buslaev and Perel'man \cite{BuPe} and that has seen an enormous activity in the last two decades. For instance, the works \cite{Cu,RoScSo} concern the situation without excited states, which is similar to the situation considered here. We will not attempt to review the immense list of works contributing to the problem with excited states.

\smallskip

A key ingredient in our proof are adiabatic dispersive estimates for time-dependent Schr\"o\-din\-ger operators which, we hope, will turn out to be useful also beyond the context of this work. They are the topic of Section \ref{sec:linear} of this paper, which can be read independently of the remaining sections. We emphasize that the notation $V$ in this part of the paper has nothing to do with the solution of \eqref{eq:eqref}.

\begin{assumption}\label{ass:disp}
Let $T\in(0,\infty)$ and let $V\in C([0,T];\langle x\rangle^{-2} L^1(\R))\cap C^1([0,T];L^1+L^\infty(\R))$ such that $V(t)$ is real-valued for any $t\in[0,T]$. Moreover, for any $t\in[0,T]$, the operator $-\partial_x^2+V(t)$ has a single negative eigenvalue and no zero energy resonance.
\end{assumption}

We denote by $P_c(t)$ the orthogonal projection corresponding to the continuous spectrum of $-\partial_x^2+V(t)$ in $L^2(\R)$ and consider the equation
\begin{equation}
\label{eq:linear}
i\epsilon\partial_t\psi = (-\partial_x^2 +V(t))P_c(t)\psi
\end{equation}
with an initial condition $\psi_0$ corresponding to the continuous spectrum of $-\partial_x^2+V(0)$.

\begin{theorem}\label{dispersive}
If $P_c(0)\psi_0=\psi_0$, then the solution $\psi$ of \eqref{eq:linear} with initial condition $\psi|_{t=0}=\psi_0$ satisfies for all $t\in (0,T]$ and all $\epsilon\in(0,1]$,
\begin{align}
\label{eq:dispthm1}
\left\| \langle x \rangle^{-1} P_c(t)\psi(t) \right\|_\infty & \lesssim \min\left\{ \left(\frac{\epsilon}{t}\right)^{1/2}, \left(\frac{\epsilon}{t}\right)^{3/2} \right\} \left\| \langle x\rangle \psi_0 \right\|_1 \,, \\
\label{eq:dispthm2}
\left\| \langle x \rangle^{-1} P_c(t)\psi(t) \right\|_\infty & \lesssim \left(\frac{\epsilon}{t}\right)^{1/2} \left\| \psi_0 \right\|_1 \,, \\
\label{eq:dispthm3}
\left\| P_c(t)\psi(t) \right\|_\infty & \lesssim \max\left\{ \left(\frac{\epsilon}{t}\right)^{1/2}, 1 \right\} \left\| \psi_0 \right\|_1 \,.
\end{align}
\end{theorem}

For us, the most important one of these bounds is \eqref{eq:dispthm1}, which yields an integrable $t^{-3/2}$ decay at the expense of introducing weights into the norms. However, we also need the bounds \eqref{eq:dispthm2} and \eqref{eq:dispthm3} without weights on the right side when dealing with some remainder terms.

The improved bound \eqref{eq:dispthm1} relies fundamentally on the non-resonance assumption on $-\partial_x^2+V(t)$. In the context of asymptotic stability of ground states for the non-linear Schr\"odinger equation, the observation that a non-resonance condition improves the usual $t^{-1/2}$ decay to a $t^{-3/2}$ decay is due to Buslaev and Perel'man \cite{BuPe} and has been used in many works thereafter, see, e.g., \cite{KrSc,GaSi}. The bounds in Theorem \ref{dispersive} for time-dependent $V$ seem to be new, but as an input in the proof we use bounds for time-independent $V$. Such bounds go back to Weder \cite{We} and are due to Goldberg and Schlag \cite{GoSc} and Mizutani \cite{Mi} under rather minimal assumption decay conditions on $V$. For further references we refer to the review \cite{Sc}. For dispersive estimate for Schr\"odinger operators with time-dependent potentials in a non-adiabatic setting in the three-dimensional case we refer to \cite{RoSc}.

\smallskip

The research described in this paper was finished in early 2017 and the results were presented at conferences in Stuttgart, Oberwolfach and Munich between April and June 2017 and announced in \cite{Fr}. In April 2019 the authors received a preprint by Leopold, Rademacher, Schlein and Seiringer \cite{LeRaScSe} which contains closely related results for the corresponding three-dimensional system, obtained by different means. Let us compare their work with ours. The techniques from \cite{LeRaScSe} extend immediately to the one-dimensional case considered here, but it is not clear whether our techniques extend to the three-dimensional case. While the dispersion in three dimensions is stronger, which would lead to some simplifications in our approach, the corresponding Schr\"odinger operator in three dimensions has typically infinitely many negative eigenvalues, which is probably outside of the scope of our methods.

The assertions in \cite{LeRaScSe}, translated into the one-dimensional setting, are different from ours. In \cite{LeRaScSe} $\psi_t$ is compared with the ground state of $-\partial_x^2+\phi_t$ (multiplied by a suitable phase), which still depends on $\epsilon$. On the other hand, our comparison dynamics $(Q,V)$ are independent of $\epsilon$ (again, up to an explicit phase). Moreover, for times of order one our bound on the approximation error for $\psi$ in $L^2$ is of order $\epsilon$ whereas it is only of order $\sqrt\epsilon$ in \cite{LeRaScSe}. We have stated our bounds only up to times of order one, even when $T^*=\infty$. In contrast, the bounds in \cite{LeRaScSe} are possibly valid, with a worse error bound, up to times of order $o(\epsilon^{-1})$, provided a certain spectral assumption is satisfied. This assumption is only verified up to times of order one. The problem of approximating $\phi_t$ is not considered in \cite{LeRaScSe}.

Finally, \cite{LeRaScSe} contains results about the relation between the classical and the quantum model, which we did not study in this paper. For earlier results about the relation between the classical and quantum dynamics we refer to \cite{FrSc,FrGa,Gr}.


\section{Dispersive estimates with time-dependent potentials}\label{sec:linear}

Our goal in this section is to prove Theorem \ref{dispersive}.


\subsection{Preparations for the proof}\label{sec:dispprep}

We denote by $\phi(t)$ an $L^2$-normalized eigenfunction corresponding to the unique negative eigenvalue of $-\partial_x^2+V(t)$ and set
$$
P_d(t):=1-P_c(t) = |\phi(t)\rangle\langle\phi(t)| \,.
$$
The second equality follows from Assumption \ref{ass:disp}. Under our assumptions on $V$, it is well-known that the eigenfunctions can be chosen to satisfy $\phi\in C^1([0,T],H^1(\R))$. In fact, in our situation, where $\phi(t)$ corresponds to the lowest eigenvalue, such a choice is fixed by requiring that $\phi(t)$ is non-negative for any $t\in[0,T]$. In the following it is only important that $\phi(t)$ is real-valued which, since $\|\phi(t)\|_2=1$ implies that $\langle \phi(t),\partial_t\phi(t)\rangle =0$.

We will frequently use the following properties of these eigenfunctions,
$$
\| \langle x\rangle \phi(t) \|_1 \,, \| \langle x\rangle^{-1} \phi(t)\|_\infty \,, \|\langle x\rangle \partial_t\phi(t) \|_1 \lesssim 1 \,.
$$
The uniform boundedness of the first two norms follows from the fact that $\phi(t)$ satisfies pointwise exponential bounds. Those follow, for instance, by writing the equation for $\phi(t)$ as a Volterra equation and using the fact that $V\in L^\infty([0,T],L^1(\R))$ and that the eigenvalue stays away from zero; see, e.g., \cite[Chapter 5]{Ya}. The uniform boundedness of the third norm follows by differentiating the equation for $\phi(t)$ with respect to $t$. Again using ODE techniques, it is easy to see that $\partial_t\phi$ satisfies pointwise exponential bounds (more precisely, it behaves like an exponential possibly multiplied by a linearly growing factor).

In order to prove Theorem \ref{dispersive} we will use Duhamel's formula in the following form, where we abbreviate
$$
\tilde\psi(t) = P_c(t)\psi(t) \,.
$$

\begin{lemma}\label{duhamel}
For all $t,t_0\in[0,T]$,
\begin{align}
\label{eq:duhamel}
\tilde\psi(t) & = e^{-i(-\partial_x^2+V(t_0))t/\epsilon} P_c(t_0)\psi_0 + P_d(t_0)\tilde\psi(t) \notag \\
& \quad + \frac{1}{i\epsilon}\int_{0}^t e^{-i(-\partial_x^2+V(t_0))(t-s)/\epsilon} P_c(t_0)(V(s)-V(t_0))\tilde\psi(s)\,ds \notag \\
& \quad - \int_{0}^t \langle\partial_s\phi(s),\tilde\psi(s)\rangle e^{-i(-\partial_x^2+V(t_0))(t-s)/\epsilon} P_c(t_0)\phi(s)\,ds \notag \\
& \quad - \int_{0}^t \int_0^s \langle \partial_{s_1}\phi(s_1),\tilde\psi(s_1)\rangle \,ds_1\, e^{-i(-\partial_x^2+V(t_0))(t-s)/\epsilon} P_c(t_0) \partial_s\phi(s)\,ds \,.
\end{align}
\end{lemma}

\begin{proof}
We first prove that for all $t,t_0,t_*\in[0,T]$,
\begin{align}
\label{eq:duhamel0}
\tilde\psi(t) & = e^{-i(-\partial_x^2+V(t_0))(t-t_*)/\epsilon} \tilde\psi(t_*) \notag \\
& \quad + \frac{1}{i\epsilon}\int_{t_*}^t e^{-i(-\partial_x^2+V(t_0))(t-s)/\epsilon} (V(s)-V(t_0))\tilde\psi(s)\,ds \notag \\
& \quad - \int_{t_*}^t \langle\partial_s\phi(s),\tilde\psi(s)\rangle e^{-i(-\partial_x^2+V(t_0))(t-s)/\epsilon}\phi(s)\,ds \notag \\
& \quad - \int_{t_*}^t \int_0^s \langle \partial_{s_1}\phi(s_1),\tilde\psi(s_1)\rangle \,ds_1\, e^{-i(-\partial_x^2+V(t_0))(t-s)/\epsilon} \partial_s\phi(s)\,ds \,.
\end{align}

Since $\partial_t P_c(t) = -\partial_t P_d(t)= - |\phi(t)\rangle\langle\partial_t\phi(t)|-|\partial_t\phi(t)\rangle\langle\phi(t)|$, the equation for $\tilde\psi$ reads
$$
i\epsilon\partial_t \tilde\psi = (-\partial_x^2 +V(t))\tilde\psi 
- i\epsilon \langle\partial_t\phi,\tilde\psi\rangle\phi 
-i\epsilon \langle \phi,\psi\rangle \partial_t\phi \,.
$$
(Here we also used the fact that $\langle\partial_t\phi,\psi\rangle= \langle\partial_t\phi,\tilde\psi\rangle$, since $\|\phi\|^2=1$ implies $\langle\phi,\partial_t\phi\rangle=0$, that is, $P_c\partial_t\phi=\partial_t\phi$.) Therefore, by Duhamel's formula,
\begin{align*}
\tilde\psi(t) & = e^{-i(-\partial_x^2+V(t_0))(t-t_*)/\epsilon} \tilde\psi(t_*) + \frac{1}{i\epsilon}\int_{t_*}^t e^{-i(-\partial_x^2+V(t_0))(t-s)/\epsilon} (V(s)-V(t_0))\tilde\psi(s)\,ds \\
& \quad - \int_{t_*}^t \langle\partial_s\phi(s),\tilde\psi(s)\rangle e^{-i(-\partial_x^2+V(t_0))(t-s)/\epsilon}\phi(s)\,ds \\
& \quad - \int_{t_*}^t \langle \phi(s),\psi(s)\rangle e^{-i(-\partial_x^2+V(t_0))(t-s)/\epsilon}\partial_s\phi(s)\,ds \,.
\end{align*}
In order to replace $\psi$ in the last integral by $\tilde\psi$ we note that
\begin{align*}
\frac{d}{dt} \langle \phi(t),\psi(t)\rangle & = \langle \partial_t\phi(t),\psi(t)\rangle + \langle \phi(t),\dot\psi(t)\rangle \\
& = \langle \partial_t\phi(t),\psi(t)\rangle + \frac{1}{i\epsilon} \langle\phi(t),(-\partial_x^2+V(t))P_c(t)\psi(t)\rangle \\
& = \langle \partial_t\phi(t),\psi(t)\rangle \\
& = \langle \partial_t\phi(t),\tilde\psi(t)\rangle \,.
\end{align*}
Thus, recalling also $P_c(0)\psi_0=\psi_0$,
\begin{equation*}
\langle \phi(t),\psi(t)\rangle = \int_0^t \langle \partial_s\phi(s),\tilde\psi(s)\rangle \,ds \,.
\end{equation*}
Inserting this into the above formula we finally obtain \eqref{eq:duhamel0}.

In order to prove the equality in the lemma, we choose $t_0=t_*$ in \eqref{eq:duhamel0} and apply $P_c(t_0)$ to obtain
\begin{align*}
\tilde\psi(t) & = P_d(t_0)\tilde\psi(t) + e^{-i(-\partial_x^2+V(t_0))(t-t_0)/\epsilon} \tilde\psi(t_0) \notag \\
& \quad + \frac{1}{i\epsilon}\int_{t_0}^t e^{-i(-\partial_x^2+V(t_0))(t-s)/\epsilon} P_c(t_0)(V(s)-V(t_0))\tilde\psi(s)\,ds \notag \\
& \quad - \int_{t_0}^t \langle\partial_s\phi(s),\tilde\psi(s)\rangle e^{-i(-\partial_x^2+V(t_0))(t-s)/\epsilon} P_c(t_0)\phi(s)\,ds \notag \\
& \quad - \int_{t_0}^t \int_0^s \langle \partial_{s_1}\phi(s_1),\tilde\psi(s_1)\rangle \,ds_1\, e^{-i(-\partial_x^2+V(t_0))(t-s)/\epsilon}P_c(t_0)\partial_s\phi(s)\,ds \,.
\end{align*}
On the other hand, taking $t_*=0$ and $t=t_0$ in \eqref{eq:duhamel0} and applying $e^{-i(-\partial_x^2+V(t_0))(t-t_0)/\epsilon} P_c(t_0)$ to both sides of the equation we obtain
\begin{align*}
& e^{-i(-\partial_x^2+V(t_0))(t-t_0)/\epsilon}\tilde\psi(t_0)
= e^{-i(-\partial_x^2+V(t_0))t/\epsilon}P_c(t_0)\psi_0 \notag \\
& \qquad + \frac{1}{i\epsilon} \int_0^{t_0} e^{-i(-\partial_x^2+V(t_0))(t-s)/\epsilon} P_c(t_0) \left( V(s)-V(t_0)\right)\tilde\psi(s)\,ds \notag \\
& \qquad - \int_0^{t_0} \langle\partial_s\phi(s),\tilde\psi(s)\rangle e^{-i(-\partial_x^2+V(t_0))(t-s)/\epsilon} P_c(t_0) \phi(s)\,ds \notag \\
& \qquad - \int_0^{t_0} \int_0^s \langle\partial_{s_1}\phi(s_1),\tilde\psi(s_1)\rangle\,ds_1 e^{-i(-\partial_x^2+V(t_0))(t-s)/\epsilon} P_c(t_0)\partial_s\phi(s)\,ds \,.
\end{align*}
Combining the previous two formulas we arrive at the claimed expression \eqref{eq:duhamel}.
\end{proof}

The following simple bounds will be useful in the proof of Theorem \ref{dispersive}.

\begin{lemma}\label{intbounds}
We have for all $\epsilon>0$, $t>0$ and $T>0$,
\begin{align*}
& \int_0^\infty  \min\left\{ \left| \frac{\epsilon}{t-s}\right|^{1/2},\left| \frac{\epsilon}{t-s}\right|^{3/2} \right\} \min\left\{ \left( \frac{\epsilon}{s}\right)^{1/2},\left( \frac{\epsilon}{s}\right)^{3/2} \right\} ds  \lesssim \epsilon \min\left\{ \left(\frac{\epsilon}{t}\right)^{1/2}, \left( \frac{\epsilon}{t} \right)^{3/2} \right\},
\end{align*}
\begin{align*}
& \int_0^\infty  \min\left\{ \left| \frac{\epsilon}{t-s}\right|^{1/2},\left| \frac{\epsilon}{t-s}\right|^{3/2} \right\} \int_0^s \min\left\{ \left( \frac{\epsilon}{s_1}\right)^{1/2},\left( \frac{\epsilon}{s_1}\right)^{3/2} \right\} ds_1 \,ds \lesssim \epsilon^2 \,,
\end{align*}
\begin{align*}
& \int_0^\infty  \min\left\{ \left| \frac{\epsilon}{t-s}\right|^{1/2},\left| \frac{\epsilon}{t-s}\right|^{3/2} \right\} \left( \frac{\epsilon}{s}\right)^{1/2} ds \lesssim \epsilon \left(\frac{\epsilon}{t}\right)^{1/2},
\end{align*}
and
\begin{align*}
& \int_0^T  \min\left\{ \left| \frac{\epsilon}{t-s}\right|^{1/2},\left| \frac{\epsilon}{t-s}\right|^{3/2} \right\} \int_0^s \left( \frac{\epsilon}{s_1}\right)^{1/2} ds_1 \,ds \lesssim \epsilon^2 \left( \frac{T}{\epsilon}\right)^{1/2}.
\end{align*}
\end{lemma}

\begin{proof}
By scaling we may assume in the following that $\epsilon=1$. To prove the first inequality we split the integral into the regions $s\leq t$ and $s>t$. For the first integral we have
\begin{align*}
& \int_0^t \min\left\{ \frac{1}{|t-s|^{1/2}}, \frac{1}{|t-s|^{3/2}} \right\} \min\left\{\frac{1}{s^{1/2}},\frac{1}{s^{3/2}} \right\} ds \\
& \quad = 2 \int_0^{t/2} \min\left\{ \frac{1}{(t-s)^{1/2}}, \frac{1}{(t-s)^{3/2}} \right\} \min\left\{\frac{1}{s^{1/2}},\frac{1}{s^{3/2}} \right\} ds \\
& \quad \lesssim \min\left\{ \frac{1}{t^{1/2}}, \frac{1}{t^{3/2}} \right\} \int_0^{t/2} \min\left\{\frac{1}{s^{1/2}},\frac{1}{s^{3/2}} \right\} ds \lesssim \min\left\{ \frac{1}{t^{1/2}}, \frac{1}{t^{3/2}}\right\}.
\end{align*}
For the second integral we have
\begin{align*}
& \int_t^\infty \min\left\{ \frac{1}{|t-s|^{1/2}}, \frac{1}{|t-s|^{3/2}} \right\} \min\left\{\frac{1}{s^{1/2}},\frac{1}{s^{3/2}} \right\} ds \\
& \quad \leq \min\left\{ \frac{1}{t^{1/2}}, \frac{1}{t^{3/2}}\right\}
\int_t^\infty \min\left\{ \frac{1}{(s-t)^{1/2}}, \frac{1}{(s-t)^{3/2}} \right\} ds
\lesssim \min\left\{ \frac{1}{t^{1/2}}, \frac{1}{t^{3/2}}\right\}.
\end{align*}
This proves the first inequality.

The second inequality simply follows from
$$
\int_0^{s} \min\left\{\frac{1}{s_1^{1/2}},\frac{1}{s_1^{3/2}} \right\}ds_1 \lesssim 1
$$
and
\begin{equation}
\label{eq:finiteintegral2}
\int_0^{\infty}  \min\left\{ \frac{1}{|t-s|^{1/2}}, \frac{1}{|t-s|^{3/2}} \right\} ds
\leq \int_\R  \min\left\{ \frac{1}{|t-s|^{1/2}}, \frac{1}{|t-s|^{3/2}} \right\} ds <\infty \,.
\end{equation}

To prove the third inequality we split the integral into the regions $s\leq t/2$ and $s>t/2$. For the first integral we have
\begin{align*}
& \int_0^{t/2} \min\left\{ \frac{1}{|t-s|^{1/2}}, \frac{1}{|t-s|^{3/2}} \right\} \frac{1}{s^{1/2}}\, ds  \lesssim \min\left\{ \frac{1}{t^{1/2}}, \frac{1}{t^{3/2}} \right\} \int_0^{t/2} \frac{1}{s^{1/2}}\, ds \lesssim \min\left\{ 1, \frac{1}{t}\right\}.
\end{align*}
For the second integral we have
\begin{align*}
& \int_{t/2}^\infty \min\left\{ \frac{1}{|t-s|^{1/2}}, \frac{1}{|t-s|^{3/2}} \right\} \frac{1}{s^{1/2}}\, ds \leq \frac{1}{t^{1/2}}
\int_{t/2}^\infty \min\left\{ \frac{1}{|t-s|^{1/2}}, \frac{1}{|t-s|^{3/2}} \right\} ds
 \lesssim \frac{1}{t^{1/2}}\,.
\end{align*}
This proves the first inequality.

The fourth inequality simply follows from
$$
\int_0^{s} \frac{1}{s_1^{1/2}}\,ds_1 = 2s^{1/2} \leq 2 T^{1/2}
$$
for $s\leq T$ and \eqref{eq:finiteintegral2}. This proves the lemma.
\end{proof}


\subsection{Proof of Theorem \ref{dispersive}. First part.}

We introduce the quantity
$$
\mathcal M(t_0,t) := \sup_{t_0\leq s\leq t} \left( \min\left\{\left(\frac{\epsilon}{s}\right)^{1/2}, \left(\frac{\epsilon}{s}\right)^{3/2}\right\} \right)^{-1} \left\|\langle x\rangle^{-1} \tilde\psi(s)\right\|_\infty
$$
and abbreviate
$$
\mathcal M(t) := \mathcal M(0,t) \,.
$$
We will show that there is a $\delta>0$ such that for all $t_0\in[0,T-\delta]$ and all $\epsilon\in(0,1]$ one has
\begin{align}\label{eq:dispersive}
& \mathcal M(t_0,t_0+\delta) \lesssim \mathcal M(t_0) + \left\| \langle x\rangle \psi_0 \right\|_1
\end{align}
with the convention that $\mathcal M(0)=0$.

Clearly, applying \eqref{eq:dispersive} iteratively at $t_0=0,\delta,2\delta,\ldots$ we obtain inequality \eqref{eq:dispthm1}.

Thus, let $0\leq t_0\leq t\leq T$. All implied constants below are independent of $t$ and $t_0$. Our starting point is the Duhamel formula \eqref{eq:duhamel}. Using the dispersive estimate in \cite{Mi} (combined with that in \cite{GoSc}) we obtain
\begin{align*}
& \left\| \langle x \rangle^{-1} \tilde\psi(t) \right\|_\infty \lesssim \min\left\{ \left( \frac{\epsilon}{t}\right)^{1/2},\left( \frac{\epsilon}{t}\right)^{3/2} \right\} \left\|\langle x\rangle \psi_0 \right\|_1 + \left\| \langle x \rangle^{-1} P_d(t_0)\tilde\psi(t) \right\|_\infty \\
& \qquad\quad + \epsilon^{-1} \int_0^t \min\left\{\left(\frac{\epsilon}{t-s}\right)^{1/2}, \left(\frac{\epsilon}{t-s}\right)^{3/2}\right\} \left\|\langle x\rangle (V(s)-V(t_0))\tilde\psi(s)\right\|_1 ds \\
& \qquad\quad + \int_0^{t_0} \left|\langle\partial_s\phi(s),\tilde\psi(s)\rangle \right| \min\left\{\left(\frac{\epsilon}{t-s}\right)^{1/2}, \left(\frac{\epsilon}{t-s}\right)^{3/2}\right\} \left\|\langle x\rangle \phi(s)\right\|_1 ds \\
& \qquad\quad + \int_{t_0}^t \left|\langle\partial_s\phi(s),\tilde\psi(s)\rangle \right| \min\left\{\left(\frac{\epsilon}{t-s}\right)^{1/2}, \left(\frac{\epsilon}{t-s}\right)^{3/2}\right\} \left\|\langle x\rangle P_c(t_0) \phi(s)\right\|_1 ds \\
& \qquad\quad + \int_{0}^t \int_0^s \left|\langle \partial_{s_1}\phi(s_1),\tilde\psi(s_1)\rangle\right| ds_1\, \min\left\{\left(\frac{\epsilon}{t-s}\right)^{1/2}\!\!, \left(\frac{\epsilon}{t-s}\right)^{3/2}\right\} \left\| \langle x\rangle \partial_s\phi(s)\right\|_1 ds.
\end{align*}
We treat the six terms on the right side separately.

The first term on the right side is already of the desired form.

To bound the second term we use the fact that $P_d(t)\tilde\psi(t)=0$ and obtain
\begin{align*}
\left\|\langle x\rangle^{-1} P_d(t_0) \tilde\psi(t)\right\|_\infty
& = \left\|\langle x\rangle^{-1} \left(P_d(t_0)-P_d(t)\right) \tilde\psi(t)\right\|_\infty \\
& \leq \left\|\langle x\rangle^{-1} \left(P_d(t_0)-P_d(t)\right)\langle x\rangle \right\|_{\infty\to\infty} \left\| \langle x\rangle^{-1} \tilde\psi(t)\right\|_\infty \,.
\end{align*}
We write
$$
P_d(t_0)-P_d(t) = |\phi(t_0)-\phi(t)\rangle\langle\phi(t_0)| + |\phi(t)\rangle\langle\phi(t_0)-\phi(t)| 
$$
and use the general fact that $\left\| |f\rangle\langle g| \right\|_{\infty\to\infty} = \|f\|_\infty \|g\|_1$ to bound
\begin{align*}
\left\|\langle x\rangle^{-1} P_d(t_0) \tilde\psi(t)\right\|_\infty
& \leq \left( \left\| \langle x\rangle^{-1}\left(\phi(t_0)-\phi(t)\right)\right\|_\infty \left\|\langle x\rangle \phi(t_0)\right\|_1 \right. \\
& \qquad \left. + \left\| \langle x \rangle^{-1} \phi(t) \right\|_\infty \left\| \langle x\rangle \left(\phi(t_0)-\phi(t)\right)\right\|_1 \right) \left\| \langle x\rangle^{-1} \tilde\psi(t)\right\|_\infty \notag \\
& \lesssim \eta(t_0,t) \left\| \langle x\rangle^{-1} \tilde\psi(t)\right\|_\infty
\end{align*}
with
\begin{align*}
\eta(t_0,t) & := \sup_{t_0\leq s\leq t} \left\|\langle x\rangle^2 (V(s)-V(t_0))\right\|_1 + \sup_{t_0\leq s\leq t} \left\|\langle x \rangle^{-1} (\phi(t_0)-\phi(s)) \right\|_\infty \\
& \qquad + \sup_{t_0\leq s\leq t} \left\|\langle x \rangle (\phi(t_0)-\phi(s)) \right\|_1 \,.
\end{align*}

For later purposes we also record the bound
\begin{align}
\label{eq:dispfirstboundproof}
\left\| \langle x \rangle P_c(t_0)\phi(s)\right\|_1 
\lesssim
\eta(t_0,t) 
\qquad\text{for all}\ t_0\leq s\leq t \,,
\end{align}
which is proved in a similar way. Indeed, since $P_c(s)\phi(s)=0$,
\begin{align*}
\left\| \langle x \rangle P_c(t_0)\phi(s)\right\|_1 & = \left\| \langle x \rangle (P_c(t_0)-P_c(s)) \phi(s)\right\|_1 = \left\| \langle x \rangle (P_d(t_0)-P_d(s)) \phi(s) \right\|_1 \\
& \leq \left\| \langle x \rangle (P_d(t_0)-P_d(s)) \langle x \rangle^{-1}\right\|_{1\to 1} \|\langle x \rangle \phi(s)\|_1 \,.
\end{align*}
We write $P_d(t_0)-P_d(s)$ as before and use the general fact that $\left\| |f\rangle\langle g| \right\|_{1\to1} = \|f\|_1 \|g\|_\infty$ to bound
\begin{align*}
\left\| \langle x \rangle P_c(t_0)\phi(s)\right\|_1
& \leq \left( \left\| \langle x \rangle (\phi(t_0)-\phi(s))\right\|_1 \left\| \langle x\rangle^{-1} \phi(t_0)\right\|_\infty \right. \\
& \qquad \left. + \left\| \langle x\rangle \phi(t) \right\|_1 \left\| \langle x\rangle^{-1} (\phi(t_0)-\phi(s)) \right\|_\infty \right) \|\langle x \rangle \phi(s)\|_1 \,,
\end{align*}
which implies \eqref{eq:dispfirstboundproof}.

To bound the third term we estimate
$$
\left\|\langle x\rangle (V(s)-V(t_0))\tilde\psi(s)\right\|_1 \leq \left\|\langle x\rangle^2 (V(s)-V(t_0))\right\|_1 \left\|\langle x\rangle^{-1} \tilde\psi(s)\right\|_\infty \,.
$$
Moreover,
\begin{equation}
\label{eq:dispproof}
\left\|\langle x\rangle^{-1} \tilde\psi(s)\right\|_\infty
\lesssim
\min\left\{ \left( \frac{\epsilon}{s}\right)^{1/2}, \left( \frac{\epsilon}{s}\right)^{3/2} \right\} \times
\begin{cases}
\mathcal M(t_0) & \text{if}\ 0\leq s\leq t_0 \,,\\
\mathcal M(t_0,t) & \text{if}\ t_0<s\leq t \,.
\end{cases}
\end{equation}
and
\begin{align*}
\left\|\langle x\rangle^2 (V(s)-V(t_0))\right\|_1
\leq
\begin{cases}
2 \sup_{0\leq s\leq T} \left\|\langle x\rangle^2 V(s) \right\|_1 \lesssim 1 & \text{if}\ 0\leq s\leq t_0 \,, \\
\eta(t_0,t) & \text{if}\ t_0<s\leq t \,.
\end{cases}
\end{align*}
Therefore, with the help of Lemma \ref{intbounds} we obtain
\begin{align*}
& \epsilon^{-1} \int_0^t \min\left\{\left(\frac{\epsilon}{t-s}\right)^{1/2}, \left(\frac{\epsilon}{t-s}\right)^{3/2}\right\} \left\|\langle x\rangle (V(s)-V(t_0))\tilde\psi(s)\right\|_1 \,ds \\
& \qquad \leq \min\left\{ \left( \frac{\epsilon}{t}\right)^{1/2}, \left( \frac{\epsilon}{t}\right)^{3/2} \right\} \left( \mathcal M(t_0) + \eta(t_0,t) \mathcal M(t_0,t) \right).
\end{align*}

To bound the fourth and the fifth term we estimate
\begin{equation}
\label{eq:dispproof1}
\left|\langle\partial_s\phi(s),\tilde\psi(s)\rangle \right| \leq \left\|\langle x \rangle \partial_s\phi(s) \right\|_1 \left\|\langle x \rangle^{-1} \tilde\psi(s) \right\|_\infty \lesssim \left\|\langle x \rangle^{-1} \tilde\psi(s) \right\|_\infty
\end{equation}
and then use \eqref{eq:dispproof}. Moreover, for $s\leq t_0$ we bound $\|\langle x\rangle \phi(s) \|_1\lesssim 1$, while for $s>t_0$ we use \eqref{eq:dispfirstboundproof}. Therefore, with the help of Lemma \ref{intbounds} we obtain
\begin{align*}
& \int_0^{t_0} \left|\langle\partial_s\phi(s),\tilde\psi(s)\rangle \right| \min\left\{\left(\frac{\epsilon}{t-s}\right)^{1/2}, \left(\frac{\epsilon}{t-s}\right)^{3/2}\right\} \left\|\langle x\rangle \phi(s)\right\|_1 \,ds \\
& \quad + \int_{t_0}^t \left|\langle\partial_s\phi(s),\tilde\psi(s)\rangle \right| \min\left\{\left(\frac{\epsilon}{t-s}\right)^{1/2}, \left(\frac{\epsilon}{t-s}\right)^{3/2}\right\} \left\|\langle x\rangle P_c(t_0) \phi(s)\right\|_1 \,ds \\
& \lesssim \epsilon \min\left\{ \left( \frac{\epsilon}{t}\right)^{1/2}, \left( \frac{\epsilon}{t}\right)^{3/2} \right\} \left( \mathcal M(t_0) + \eta(t_0,t) \mathcal M(t_0,t) \right).
\end{align*}

To bound the sixth term we use again \eqref{eq:dispproof1} and, for $s\leq t_0$, \eqref{eq:dispproof}. We interchange the order of integration. To the part of the double integral corresponding to $s_1\leq t_0$ we apply Lemma~\ref{intbounds}. In the part corresponding to $s_1\geq t_0$ we use
\begin{equation}
\label{eq:finiteintegral0}
\int_{s_1}^t \min\left\{\left(\frac{\epsilon}{t-s}\right)^{1/2}, \left(\frac{\epsilon}{t-s}\right)^{3/2}\right\} ds \leq \int_{-\infty}^t \min\left\{\left(\frac{\epsilon}{t-s}\right)^{1/2}, \left(\frac{\epsilon}{t-s}\right)^{3/2}\right\} ds = 4\epsilon \,.
\end{equation}
This implies
\begin{align*}
& \int_{0}^t \int_0^s \left|\langle \partial_{s_1}\phi(s_1),\tilde\psi(s_1)\rangle\right| ds_1\, \min\left\{\left(\frac{\epsilon}{t-s}\right)^{1/2}, \left(\frac{\epsilon}{t-s}\right)^{3/2}\right\} \left\| \langle x\rangle \partial_s\phi(s)\right\|_1 \,ds \\
& \qquad \lesssim \epsilon^2 \mathcal M(t_0) + \epsilon \int_{t_0}^t \left\|\langle x \rangle^{-1} \tilde\psi(s_1) \right\|_\infty ds_1
\end{align*}

To summarize, we have shown that
\begin{align*}
\left\| \langle x \rangle^{-1} \tilde\psi(t) \right\|_\infty \lesssim & \min\left\{ \left( \frac{\epsilon}{t}\right)^{1/2}, \left( \frac{\epsilon}{t}\right)^{3/2} \right\} \left(
\left\| \langle x \rangle\psi_0 \right\|_1 + \mathcal M(t_0) + \eta(t_0,t) \mathcal M(t_0,t) \right) \\
& + \epsilon \int_{t_0}^t \left\| \langle x \rangle^{-1} \tilde\psi(s_1)\right\|_\infty ds_1 + \epsilon^{2} \mathcal M(t_0) \,.
\end{align*}
Using $\min\{ (\epsilon/t)^{1/2},(\epsilon/t)^{3/2}\} \gtrsim \epsilon^{3/2}$ we find that
$$
u(t) := \left( \min\left\{\left(\frac{\epsilon}{t}\right)^{1/2}, \left(\frac{\epsilon}{t}\right)^{3/2}\right\} \right)^{-1} \left\| \langle x \rangle^{-1} \tilde\psi(t) \right\|_\infty
$$
satisfies
$$
u(t) \lesssim A(t_0,t) + \epsilon^{-1/2} \int_{t_0}^t \min\left\{\left(\frac{\epsilon}{s_1}\right)^{1/2}, \left(\frac{\epsilon}{s_1}\right)^{3/2}\right\} u(s_1)\,ds_1
$$
with
\begin{align*}
A(t_0,t) & = \left\| \langle x \rangle\psi_0 \right\|_1 + \mathcal M(t_0) + \eta(t_0,t) \mathcal M(t_0,t) \,.
\end{align*}
Thus, by Gronwall's inequality and a computation as in \eqref{eq:finiteintegral0},
$$
u(t) \lesssim A(t_0,t) e^{C \epsilon^{-1/2} \int_{t_0}^t \min\left\{\left(\epsilon/s_1\right)^{1/2}, \left(\epsilon/s_1 \right)^{3/2}\right\} \,ds_1} \lesssim A(t_0,t) \,.
$$
This implies that
$$
\mathcal M(t_0,t) \lesssim A(t_0,t) \,.
$$
Since $t\mapsto V(t)$ and $t\mapsto \phi(t)$ are uniformly continuous from the compact interval $[0,T]$ to $\langle x \rangle^{-2} L^1$ and to $\langle x\rangle L^\infty\cap \langle x\rangle^{-1} L^1$, respectively, by choosing $\delta$ small enough (independent of $t_0$), we can make $\eta(t_0,t_0+\delta)$ smaller than any given constant. Thus, the term $\eta(t_0,t)\mathcal M(t_0,t)$ in $A(t_0,t)$ can be absorbed in the left side and we obtain the claimed bound \eqref{eq:dispersive}.
\qed


\subsection{Proof of Theorem \ref{dispersive}. Second part.}
We introduce the quantity
$$
\tilde{\mathcal M}(t_0,t) := \sup_{t_0\leq s\leq t} \left(\frac{\epsilon}{s}\right)^{-1/2} \left\|\langle x\rangle^{-1} \tilde\psi(s)\right\|_\infty
$$
and abbreviate
$$
\tilde{\mathcal M}(t) := \tilde{\mathcal M}(0,t) \,.
$$
We will show that there is a $\delta>0$ such that for all $t_0\in[0,T-\delta]$ and all $\epsilon\in(0,1]$ one has
\begin{align}\label{eq:dispersive2}
& \tilde{\mathcal M}(t_0,t_0+\delta) \lesssim \tilde{\mathcal M}(t_0) + \left\| \psi_0 \right\|_1
\end{align}
with the convention that $\mathcal M(0)=0$.

Clearly, applying \eqref{eq:dispersive2} iteratively at $t_0=0,\delta,2\delta,\ldots$ we obtain inequality \eqref{eq:dispthm2}.

Thus, let $0\leq t_0\leq t\leq T$. All implied constants below are independent of $t$ and $t_0$. Our starting point is the Duhamel formula \eqref{eq:duhamel}. Using the dispersive estimate in \cite{Mi} (combined with that in \cite{GoSc}) we obtain
\begin{align*}
& \left\| \langle x \rangle^{-1} \tilde\psi(t) \right\|_\infty \lesssim \left( \frac{\epsilon}{t}\right)^{1/2} \left\| \psi_0 \right\|_1 + \left\| \langle x \rangle^{-1} P_d(t_0)\tilde\psi(t) \right\|_\infty \\
& \qquad\quad + \epsilon^{-1} \int_0^t \min\left\{\left(\frac{\epsilon}{t-s}\right)^{1/2}, \left(\frac{\epsilon}{t-s}\right)^{3/2}\right\} \left\|\langle x\rangle (V(s)-V(t_0))\tilde\psi(s)\right\|_1 ds \\
& \qquad\quad + \int_0^{t_0} \left|\langle\partial_s\phi(s),\tilde\psi(s)\rangle \right| \min\left\{\left(\frac{\epsilon}{t-s}\right)^{1/2}, \left(\frac{\epsilon}{t-s}\right)^{3/2}\right\} \left\|\langle x\rangle \phi(s)\right\|_1 ds \\
& \qquad\quad + \int_{t_0}^t \left|\langle\partial_s\phi(s),\tilde\psi(s)\rangle \right| \min\left\{\left(\frac{\epsilon}{t-s}\right)^{1/2}, \left(\frac{\epsilon}{t-s}\right)^{3/2}\right\} \left\|\langle x\rangle P_c(t_0) \phi(s)\right\|_1 ds \\
& \qquad\quad + \int_{0}^t \int_0^s \left|\langle \partial_{s_1}\phi(s_1),\tilde\psi(s_1)\rangle\right| ds_1\, \min\left\{\left(\frac{\epsilon}{t-s}\right)^{1/2}\!\!, \left(\frac{\epsilon}{t-s}\right)^{3/2}\right\} \left\| \langle x\rangle \partial_s\phi(s)\right\|_1 ds.
\end{align*}
Arguing in the same way as before, using Lemma \ref{intbounds}, we obtain
\begin{align*}
\left\| \langle x \rangle^{-1} \tilde\psi(t) \right\|_\infty \lesssim & \left( \frac{\epsilon}{t}\right)^{1/2} \left( \left\| \psi_0 \right\|_1 + \tilde{\mathcal M}(t_0) + \eta(t_0,t) \tilde{\mathcal M}(t_0,t) \right) \\
& + \epsilon \int_{t_0}^t \left\| \langle x \rangle^{-1} \tilde\psi(s_1)\right\|_\infty ds_1 + \epsilon^{3/2} \tilde{\mathcal M}(t_0)
\end{align*}
with $\eta(t_0,t)$ as before. Thus
$$
\tilde u(t) := \left(\frac{\epsilon}{t}\right)^{-1/2} \left\| \langle x \rangle^{-1} \tilde\psi(t) \right\|_\infty
$$
satisfies
$$
\tilde u(t) \lesssim \tilde A(t_0,t) + \epsilon^{1/2} \int_{t_0}^t \left(\frac{\epsilon}{s_1}\right)^{1/2} \tilde u(s_1)\,ds_1
$$
with
\begin{align*}
\tilde A(t_0,t) & = \left\| \psi_0 \right\|_1 + \tilde{\mathcal M}(t_0) + \eta(t_0,t) \tilde{\mathcal M}(t_0,t) \,.
\end{align*}
Thus, by Gronwall's inequality,
$$
u(t) \lesssim \tilde A(t_0,t) e^{C \epsilon^{1/2} \int_{t_0}^t \left(\epsilon/s_1 \right)^{1/2} \,ds_1} \lesssim \tilde A(t_0,t) \,.
$$
This implies that
$$
\tilde{\mathcal M}(t_0,t) \lesssim \tilde A(t_0,t)
$$
and the term $\eta(t_0,t)\tilde{\mathcal M}(t_0,t)$ in $\tilde A(t_0,t)$ can be absorbed, as before, in the left side by choosing $\delta$ small enough. This yields the claimed bound \eqref{eq:dispersive2}.
\qed


\subsection{Proof of Theorem \ref{dispersive}. Third part.}
Finally, we deduce \eqref{eq:dispthm3} from \eqref{eq:dispthm2}. The starting point is Duhamel's formula \eqref{eq:duhamel} with $t_0=t$. Applying the dispersive bounds from \cite{GoSc} we obtain
\begin{align*}
\left\| \tilde\psi(t) \right\|_\infty & \lesssim \left(\frac{\epsilon}{t}\right)^{1/2} \left\|\psi_0\right\|_1 \\
& \quad + \frac{1}{\epsilon} \int_0^t \left(\frac{\epsilon}{t-s}\right)^{1/2} \left\| \left( V(s)-V(t) \right) \tilde\psi(s) \right\|_1 ds \\
& \quad + \int_0^t \left\|\langle x \rangle \partial_s\phi(s)\right\|_1 \left\|\langle x \rangle^{-1}\tilde\psi(s) \right\|_\infty \left(\frac{\epsilon}{t-s}\right)^{1/2} \left\|\phi(s)\right\|_1 ds \\
& \quad + \int_0^t \int_0^s \left\|\langle x \rangle \partial_{s_1}\phi(s_1)\right\|_1 \left\|\langle x \rangle^{-1}\tilde\psi(s_1) \right\|_\infty \left(\frac{\epsilon}{t-s}\right)^{1/2} \left\|\partial_s\phi(s)\right\|_1 ds_1\,ds \\
& \lesssim \left( \left(\frac{\epsilon}{t}\right)^{1/2} + \frac{1}{\epsilon} \int_0^t \left(\frac{\epsilon}{t-s}\right)^{1/2} \left(\frac{\epsilon}{s}\right)^{1/2} ds + \int_0^t \int_0^s \left(\frac{\epsilon}{s_1}\right)^{1/2} \left(\frac{\epsilon}{t-s}\right)^{1/2} ds_1\,ds \right)\\
& \quad \times\left\|\psi_0\right\|_1 \\
& \lesssim \left( \left(\frac{\epsilon}{t}\right)^{1/2} + 1 \right) \left\|\psi_0\right\|_1 \,.
\end{align*}
In the next to last inequality we used \eqref{eq:dispthm2}. This completes the proof of the theorem.
\qed


\section{Reminder on the adiabatic theorem}\label{sec:adiabatic}

In this section we briefly recall a version of the usual adiabatic theorem. This material is well-known, but we have not been able to find a reference for the precise inequality in Theorem \ref{adiabaticenergy} that we need. Since it comes at no extra effort, we present the material in a general Hilbert space. For further results and references concerning the adiabatic theorem we refer, for instance, to \cite{Te}.

For $t\in [0,T]$ let $H(t)$ be a self-adjoint operator in a Hilbert space. We assume that for any $t\in[0,T]$, $E(t)$ is a simple eigenvalue of $H(t)$. We assume that $E(t)$ and a corresponding normalized eigenvector $\Phi(t)$ depend in a $C^2$ manner on time. (If the resolvent of $H(t)$ is $C^2$ with respect to $t$ in operator norm and if $E(t)$ is isolated in the spectrum, which we do not assume, however, this assumption is automatically satisfied.) We set
$$
\beta(t) := \int_0^t \im\langle\Phi(s),\dot\Phi(s)\rangle\,ds
\qquad\text{and}\qquad
\theta(t) := \int_0^t E(s)\,ds \,.
$$
Differentiating $\|\Phi(t)\|^2=1$ we infer that $\dot\Phi(t) - i\im\langle\Phi(t),\dot\Phi(t)\rangle \Phi(t)$ is orthogonal to $\Phi(t)$. Assuming that $\dot\Phi(t)$ belongs to the operator domain of $H(t)$, it follows that
$$
\Xi(t) := e^{-i\beta(t)} \frac{1}{H(t)-E(t)}\left(\dot\Phi(t) - i\im\langle\Phi(t),\dot\Phi(t)\rangle \Phi(t) \right)
$$
is well-defined and orthogonal to $\Phi(t)$.

The adiabatic theorem says that the solution $\Psi(t)$ to
\begin{equation}
\label{eq:schreq}
i\epsilon \partial_t\Psi(t) = H(t)\Psi(t) \,,
\qquad \Psi(0) = \Phi(0)
\end{equation}
is approximately given by $e^{-i\theta(t)/\epsilon-i\beta(t)} \Phi(t)$. The following two theorems quantify this in the norm of the underlying Hilbert space and in the `energy norm', respectively.

\begin{theorem}\label{adiabatic}
Let $\Psi$ be the solution of \eqref{eq:schreq}. Then
\begin{align*}
\left\| \Psi(t) - e^{-i\theta(t)/\epsilon-i\beta(t)} \Phi(t) \right\|
\leq 2\epsilon \sup_{0\leq s\leq t} \left( \|\Xi(s)\| + \int_0^s \left\| \dot\Xi(s_1)\right\| ds_1 \right).
\end{align*}
In particular, if $\Pi(t)=1-|\Phi(t)\rangle\langle\Phi(t)|$, then
$$
\| \Pi(t) \Psi(t) \| \leq 2\epsilon \sup_{0\leq s\leq t} \left( \|\Xi(s)\| + \int_0^s \left\| \dot\Xi(s_1)\right\| ds_1 \right).
$$
\end{theorem}

\begin{proof}
We first observe that without loss of generality we may assume that $E(t)=0$ and $\langle\Phi(t),\dot\Phi(t)\rangle=0$ for all $t$ (that is, $\theta=\beta=0$). In fact, if we have proved the theorem in this case, we can apply it to $\tilde H(t)=H(t)-E(t)$, $\tilde E(t)=0$, $\tilde\Phi(t) = e^{-i\beta(t)}\Phi(t)$ and $\tilde\Psi(t) = e^{i\theta(t)/\epsilon}\Psi(t)$ and obtain the theorem as stated.

Thus, assuming now $E(t)=\langle\Phi(t),\dot\Phi(t)\rangle=0$, we compute
\begin{align*}
\frac{d}{dt} \left\| \Psi(t) - \Phi(t) \right\|^2
& = 2 \re \left\langle \Psi(t) - \Phi(t), \dot\Psi(t) - \dot\Phi(t) \right\rangle \\
& = 2 \re \left\langle \Psi(t) - \Phi(t), \frac{1}{i\epsilon}H(t)\Psi(t) - \dot\Phi(t) \right\rangle \\
& = 2\re \langle \Psi(t), \dot\Phi(t) \rangle \,.
\end{align*}
Here we used the fact that $H(t)$ is self-adjoint. We now insert the definition of $\Xi(t)$ and obtain
$$
\langle \Psi(t), \dot\Phi(t)\rangle
= \langle H(t) \Psi(t), \Xi(t) \rangle
= -i\epsilon \langle \dot\Psi(t), \Xi(t)\rangle \,.
$$
Thus, we have shown that
$$
\frac{d}{dt} \left\| \Psi(t) - \Phi(t) \right\|^2 = 2\epsilon \im \langle \dot\Psi(t), \Xi(t)\rangle \,.
$$
Since $\langle\dot\Phi(t),\Xi(t)\rangle = \langle\dot\Phi(t),H(t)^{-1}\dot\Phi(t)\rangle$ is real, we have
\begin{align*}
\frac{d}{dt} \left\| \Psi(t) - \Phi(t) \right\|^2 & = 2\epsilon \im \langle \dot\Psi(t)-\dot\Phi(t), \Xi(t)\rangle \\
& = 2\epsilon\im \left( \frac{d}{dt}\langle\Psi(t)-\Phi(t),\Xi(t)\rangle - \langle\Psi(t)-\Phi(t),\dot\Xi(t)\rangle \right).
\end{align*}
Integrating this and recalling that $\Psi(0)=\Phi(0)$, we obtain
$$
\left\| \Psi(t) - \Phi(t) \right\|^2 = 2\epsilon\im \left( \langle\Psi(t)-\Phi(t),\Xi(t)\rangle - \int_0^t \langle\Psi(s)-\Phi(s),\dot\Xi(s)\rangle \,ds \right).
$$
Thus,
$$
\left\| \Psi(t) - \Phi(t) \right\|^2 \leq 2\epsilon \left( \|\Xi(t)\|+ \int_0^t \|\dot\Xi(s)\| \,ds \right) \sup_{0\leq s\leq t} \left\|\Psi(s)-\Phi(s)\right\|.
$$
This implies the first bound in the theorem.

To deduce the second one, we observe that $1-|z|^2\leq 2(1-\re z)$ for all $z\in\C$ with $|z|\leq 1$, and obtain, using $\|\Psi(t)\|=1$ (which follows from the self-adjointness of $H(t)$),
$$
\|\Pi(t)\Psi(t)\|^2 = 1- \left|\langle\Phi(t),\Psi(t)\rangle\right|^2 \leq 2\left( 1 - \re \langle\Phi(t),\Psi(t)\rangle \right) = \left\| \Psi(t) - \Phi(t)\right\|^2 \,.
$$
Therefore, the second assertion follows from the first one.
\end{proof}

We now assume, in addition, that there are $C_1, C_2\geq 0$ such that for all $t\in[0,T]$ and all $\psi$ in the form domain of $H(t)$,
$$
\langle \psi,\dot H(t) \psi\rangle \leq C_1 \langle \psi,H(t)\psi\rangle + C_2 \|\psi\|^2 \,.
$$

\begin{theorem}\label{adiabaticenergy}
Let $\Psi$ be the solution of \eqref{eq:schreq}. Then
$$
\left\langle \Psi(t) - e^{-i\theta(t)/\epsilon-i\beta(t)}\Phi(t), \left(H(t)-E(t)\right) \left( \Psi(t) - e^{-i\theta(t)/\epsilon-i\beta(t)}\Phi(t) \right) \right\rangle \leq \sup_{0\leq s\leq t} C(s) \ e^{C_1 t} \ \epsilon^2 \,,
$$
where
\begin{align}
\label{eq:c(t)}
C(t) & = 4 \left( \left\| \frac{d}{dt} \left( e^{-i\beta(t)} \Phi(t) \right) \right\| + \int_0^t \left\|\frac{d^2}{ds^2}\left( e^{-i\beta(s)} \Phi(s) \right) \right\|ds \right) \notag \\
& \quad\quad \times \sup_{0\leq s\leq t} \left( \left\|\Xi(s)\right\| + \int_0^s \left\|\dot\Xi(s_1)\right\|ds_1 \right) \notag \\
& \quad + 4C_2 t\sup_{0\leq s\leq t} \left( \left\|\Xi(s)\right\| + \int_0^s \left\|\dot\Xi(s_1)\right\|ds_1 \right)^2 \,.
\end{align}
In particular, if $\Pi(t)=1-|\Phi(t)\rangle\langle\Phi(t)|$, then
$$
\left\langle \Pi(t) \Psi(t), \left(H(t)-E(t)\right) \Pi(t) \Psi(t) \right\rangle \leq \sup_{0\leq s\leq t} C(s) \ e^{C_1 t} \ \epsilon^2 \,.
$$
\end{theorem}

\begin{proof}
By the same argument as in the proof of Theorem \ref{adiabatic} we may assume that $E(t)=\langle\Phi(t),\dot\Phi(t)\rangle=0$ for all $t$. Note that
$$
\left\langle\Psi(t)-\Phi(t),H(t)\left(\Psi(t)-\Phi(t)\right)\right\rangle = \left\langle\Psi(t),H(t)\Psi(t)\right\rangle,
$$
and therefore, using the self-adjointness of $H(t)$,
\begin{align*}
\frac{d}{dt}\left\langle\Psi(t)-\Phi(t),H(t)\left(\Psi(t)-\Phi(t)\right)\right\rangle 
& = 2\re \left\langle\Psi(t),H(t)\dot\Psi(t)\right\rangle + \left\langle\Psi(t),\dot H(t)\Psi(t)\right\rangle \\
& = 2\re \frac{1}{i\epsilon} \left\langle\Psi(t),H(t)^2\Psi(t)\right\rangle + \left\langle\Psi(t),\dot H(t)\Psi(t)\right\rangle \\
& = \left\langle\Psi(t),\dot H(t)\Psi(t)\right\rangle \\
& = \left\langle\Phi(t),\dot H(t)\Phi(t)\right\rangle + 2\re \left\langle\Psi(t)-\Phi(t),\dot H(t)\Phi(t)\right\rangle \\
& \qquad + \left\langle\Psi(t)-\Phi(t),\dot H(t)\left(\Psi(t)-\Phi(t)\right)\right\rangle.
\end{align*}
We now use the fact that
$$
\dot H(t)\Phi(t) = \partial_t \left( H(t)\Phi(t) \right) - H(t)\dot\Phi(t) = -H(t)\dot \Phi(t)
$$
to write
$$
\left\langle\Phi(t),\dot H(t)\Phi(t)\right\rangle = - \left\langle\Phi(t),H(t)\dot \Phi(t)\right\rangle = 0
$$
and
$$
\left\langle\Psi(t)-\Phi(t),\dot H(t)\Phi(t)\right\rangle = - \left\langle\Psi(t)-\Phi(t), H(t)\dot\Phi(t)\right\rangle = i\epsilon \left\langle\dot\Psi(t),\dot\Phi(t) \right\rangle.
$$
Thus, we have shown that
\begin{align*}
& \frac{d}{dt}\left\langle\Psi(t)-\Phi(t),H(t)\left(\Psi(t)-\Phi(t)\right)\right\rangle \\
& \qquad = -2\epsilon \im \left\langle\dot\Psi(t),\dot\Phi(t) \right\rangle + \left\langle\Psi(t)-\Phi(t),\dot H(t)\left(\Psi(t)-\Phi(t)\right)\right\rangle \\
& \qquad = -2\epsilon \im \left\langle\dot\Psi(t)-\dot\Phi(t),\dot\Phi(t) \right\rangle \\
& \qquad \quad + \left\langle\Psi(t)-\Phi(t),\dot H(t)\left(\Psi(t)-\Phi(t)\right)\right\rangle \\
& \qquad = -2\epsilon \im \left( \frac{d}{dt} \left\langle\Psi(t)-\Phi(t),\dot\Phi(t) \right\rangle - \left\langle\Psi(t)-\Phi(t),\ddot\Phi(t) \right\rangle \right) \\
& \qquad\quad + \left\langle\Psi(t)-\Phi(t),\dot H(t)\left(\Psi(t)-\Phi(t)\right)\right\rangle.
\end{align*}
Integrating and recalling that $\Psi(0)=\Phi(0)$, we obtain
\begin{align*}
& \left\langle\Psi(t)-\Phi(t),H(t)\left(\Psi(t)-\Phi(t)\right)\right\rangle \\
& \qquad = -2\epsilon \im \left( \left\langle\Psi(t)-\Phi(t),\dot\Phi(t) \right\rangle - \int_0^t \left\langle\Psi(s)-\Phi(s),\ddot\Phi(s) \right\rangle ds \right) \\
& \qquad \quad + \int_0^t \left\langle\Psi(s)-\Phi(s),\dot H(s)\left(\Psi(s)-\Phi(s)\right)\right\rangle ds\,.
\end{align*}
We now set $f(t) :=  \left\langle\Psi(t)-\Phi(t),H(t)\left(\Psi(t)-\Phi(t)\right)\right\rangle$ and bound, using the assumption on $\dot H$ and Theorem \ref{adiabatic}, 
\begin{align*}
f(t) & \leq 2\epsilon \left( \left\|\dot\Phi(t)\right\| + \int_0^t \left\|\ddot\Phi(s)\right\|ds \right) \sup_{0\leq s\leq t} \left\|\Psi(s)-\Phi(s)\right\| \\
& \quad + C_1 \int_0^t f(s)\,ds + C_2 t \sup_{0\leq s\leq t} \left\|\Psi(s)-\Phi(s)\right\|^2 \\
& \leq C(t) \epsilon^2 + C_1 \int_0^t f(s)\,ds
\end{align*}
with $C(t)$ from \eqref{eq:c(t)}. The first inequality in the theorem now follows from Gronwall's inequality.

To prove the second inequality, we simply observe that (still assuming $E=\langle\Phi,\dot\Phi\rangle=0$)
$$
\left\langle \Pi(t) \Psi(t), H(t) \Pi(t) \Psi(t) \right\rangle = \left\langle \Psi(t)-\Phi(t), H(t) \left(\Psi(t)-\Phi(t)\right) \right\rangle
$$
and apply the first bound.
\end{proof}


\section{The reference dynamics}

We now turn our attention to the non-linear dynamics. As a first step in the proof of Theorem \ref{main} we will prove the existence of a unique solution to the ($\epsilon$-independent) reference system \eqref{eq:eqref} with the mass conservation condition \eqref{eq:refmass} and the initial conditions \eqref{eq:refinitial}.

\begin{proposition}\label{reference}
Let $\phi_0,\dot\phi_0\in L^2(\R)$ be real-valued and assume that $-\partial_x^2+\phi_0$ has an eigenvalue $E_0<0$. Let $\psi_0$ be an associated real-valued eigenfunction. Then there is a maximal $T_*\in(0,\infty)\cup\{\infty\}$ and unique functions $Q\in C([0,T_*),H^1(\R,\R))$, $V\in C^1([0,T_*),L^2(\R,\R))$ and $E\in C([0,T_*),(-\infty,0))$ such that \eqref{eq:eqref}, \eqref{eq:refmass} and \eqref{eq:refinitial} hold. Moreover, we have $Q\in C^\infty([0,T_*),H^1(\R,\R))$, $V\in C^\infty([0,T_*),L^2(\R,\R))$ and $E\in C^\infty([0,T_*),(-\infty,0))$
\end{proposition}

\begin{proof}
We will prove existence and uniqueness on a small time interval $[0,\tau]$. Once this is shown, by iterating the argument we obtain existence and uniqueness on a maximal time interval. Note that, by solving the equation for $V$ together with its initial conditions we obtain
\begin{equation}
\label{eq:solv}
V_t = \phi_0 \cos t + \dot\phi_0 \sin t - \tfrac{1}{2} \int_0^t Q_s^2 \sin(t-s) \,ds\,.
\end{equation}
From this formula it is clear that if $Q$ is continuous in time, then $V$ is $C^2$ in time. Consequently, by standard perturbation theory and the fact that eigenvalues of one-dimensional Schr\"odinger operators are simple, the eigenvalue $E$ and the eigenfunction $Q$ are also $C^2$ in time. Thus, by \eqref{eq:solv}, $V$ is $C^4$ in time, and iterating the above argument we obtain the $C^\infty$ assertion in the proposition.

As another consequence of \eqref{eq:solv}, we can consider $V$ as a functional of $Q$,
$$
\mathcal V_Q(t) := \phi_0 \cos t + \dot\phi_0 \sin t - \tfrac{1}{2} \int_0^t Q_s^2 \sin(t-s) \,ds\,,
$$
and think of only $Q$ and $E$ as the unknowns. The crucial step in the proof is the following

\emph{Claim.} For any $\rho>0$ there is a $\tau>0$ such that for any $\delta\in C([0,\tau],[0,1/2])$ with $\delta(0)=0$ there are $f\in C([0,\tau],H^1(\R,\R)$ with $\langle\psi_0,f\rangle=0$ and $f_0=0$ and $e\in C([0,\tau],(-\infty,0))$ with $e_0=0$ such that
\begin{equation}
\label{eq:decompqref}
Q := (1-\delta)\psi_0 + f
\qquad\text{and}\qquad
E=E_0+e
\end{equation}
satisfy \eqref{eq:eqref} and \eqref{eq:refinitial}. Moreover,
$$
\| f\|_2 \leq \rho \|\psi_0\|_2
$$
and, emphasizing the $\delta$ dependence,
\begin{equation}
\label{eq:contractionproof2}
\sup_{0\leq t\leq\tau} \| f_t^{(\delta_1)} - f_t^{(\delta_2)} \|_2 \leq C \sup_{0\leq t\leq\tau} |\delta_1(t)-\delta_2(t)|
\end{equation}
with a universal constant $C$, independent of $\rho$.

Let us accept this claim for the moment and complete the proof. The idea is to determine $\delta$ so as to satisfy \eqref{eq:refmass}. The latter is equivalent to $\|\psi_0\|_2^2 = (1-\delta)^2\|\psi_0\|_2^2 + \|f^{(\delta)}\|_2^2$ and therefore to a fixed point of
$$
F(\delta)(t):= 1-\sqrt{1-\|f^{(\delta)}_t\|_2^2/\|\psi_0\|_2^2}
$$
defined on $\mathcal X=\{ \delta\in C([0,\tau],[0,1/2]):\ \delta(0)=0\}$. Here $\tau$ will be chosen later as in the claim, depending on $\rho$.

Since $f_0=0$ we have $F(\delta)(0)=0$ and, since $0\leq 1-\sqrt{1-x}\leq x$ for all $0\leq x\leq 1$, we have
$$
0 \leq F(\delta) \leq \|f^{(\delta)}_t\|_2^2/\|\psi_0\|_2^2 \leq \rho^2 \,,
$$
so $F(\delta)\in\mathcal X$ provided $\rho^2\leq 1/2$, which we assume in the following. Moreover,
\begin{align*}
|F(\delta_1)-F(\delta_2)| & = \|\psi_0\|_2^{-2} \frac{\left|\| f^{(\delta_1)}\|_2^2 - \|f^{(\delta_2)}\|_2^2 \right|}{\sqrt{1-\|f^{(\delta_1)}\|_2^2/\|\psi_0\|_2^2}+\sqrt{1-\|f^{(\delta_2)}\|_2^2/\|\psi_0\|_2^2}} \\
& \leq \frac{\rho}{\sqrt{1-\rho^2}} \|\psi_0\|_2^{-1} \| f^{(\delta_1)}- f^{(\delta_2)} \|_2 \\
& \leq \frac{C\rho}{\sqrt{1-\rho^2}} \|\psi_0\|_2^{-1} \sup_{0\leq t\leq\tau} |\delta_1(t)-\delta_2(t)| \,.
\end{align*}
The last inequality is valid provided $\tau$ is chosen as in the claim depending on $\rho$ such that \eqref{eq:contractionproof2} holds. Choosing $\rho>0$ sufficiently small we see that $F$ is a contraction in $\mathcal X$. Therefore $F$ has a unique fixed point in $\mathcal X$, as stated in the proposition.

It remains to verify the claim, which we will do by another fixed point argument. Inserting the decomposition \eqref{eq:decompqref} into \eqref{eq:eqref} we obtain
$$
(-\partial_x^2+\phi_0-E_0)f = -(\mathcal V_{(1-\delta)Q+f}-\phi_0)((1-\delta)\psi_0+f) + e((1-\delta)\psi_0+f) \,.
$$
Projecting this equation onto the span of $\psi_0$ and its orthogonal complement, we see that the equation is equivalent to a fixed point of the map
$$
F(e,f):=
\begin{pmatrix}
(1-\delta)^{-1} \|\psi_0\|_2^{-1} \langle\psi_0, (\mathcal V_{(1-\delta)Q+f}-\phi_0)((1-\delta)\psi_0+f)\rangle \\
- R (\mathcal V_{(1-\delta)Q+f}-\phi_0)((1-\delta)\psi_0+f) + e Rf
\end{pmatrix}.
$$
Here $R$ is the inverse of $-\partial_x^2+\phi_0-E_0$ defined on the orthogonal complement of $\psi_0$. Since $e_0$ is a simple isolated eigenvalue, $R$ is a bounded operator from $H^{-1}(\R)$ to $H^1(\R)$. We will consider $F$ as a map on $\mathcal Y\times \mathcal Z$, where
\begin{align*}
\mathcal Y & = \left\{ e\in C([0,\tau],\R):\ e(0)=0 \,,\ \sup_{0\leq t\leq\tau} |e_t|\leq \sigma \right\}, \\
\mathcal Z & = \left\{ f\in C([0,\tau],H^1(\R,\R)):\ \langle\psi_0,f\rangle=0 \,, \sup_{0\leq t\leq\tau} \|f_t\|_{H^1} \leq \sigma \right\}.
\end{align*}
Here $\tau>0$ and $\sigma>0$ will be chosen later sufficiently small. In particular, we will choose $\sigma\leq \rho\|\psi_0\|_2$.

It follows from \eqref{eq:solv} that for $f\in\mathcal Z$
\begin{equation}
\label{eq:contractionproof1}
\| \mathcal V_{(1-\delta)Q+f}(t) - \phi_0 \|_2 \lesssim t \,.
\end{equation}
Using the Sobolev embedding $H^1\subset L^4$ and, by duality, $L^{4/3}\subset H^{-1}$, we obtain
$$
\| (\mathcal V_{(1-\delta)Q+f}-\phi_0)((1-\delta)\psi_0+f) \|_{H^{-1}} \lesssim \|\mathcal V_{(1-\delta)Q+f}-\phi_0 \|_2 \|(1-\delta)\psi_0+f \|_4 \lesssim t
$$
and, consequently, for the two components $F_1$ and $F_2$ of the map $F$,
$$
\sup_{0\leq t\leq\tau} | F_1(e,f)(t) | \lesssim \tau \,,
\qquad
\sup_{0\leq t\leq\tau} \| F_2(e,f)(t) \|_{H^1} \lesssim \tau + \sigma^2
$$
Thus, if $\tau$ is small compared to $\sigma$ and $\sigma$ is small compared to $1$, then $F$ maps $\mathcal Y\times\mathcal Z$ into itself.

To prove the contraction property and already preparing for the proof of \eqref{eq:contractionproof2}, we want to bound, similarly as before, the $L^{4/3}$ norm of
\begin{align*}
& (\mathcal V_{(1-\delta_1)Q+f_1}-\phi_0)((1-\delta_1)\psi_0+f_1) - (\mathcal V_{(1-\delta_2)Q+f_2}-\phi_0)((1-\delta_2)\psi_0+f_2) \\
& = \left(\mathcal V_{(1-\delta_1)Q+f_1} - \mathcal V_{(1-\delta_2)Q+f_2} \right) \left( (1-\tfrac12(\delta_1+\delta_2))\psi_0 + \tfrac12(f_1+f_2) \right) \\
& \quad + \left( \tfrac12 \left( \mathcal V_{(1-\delta_1)Q+f_1} + \mathcal V_{(1-\delta_2)Q+f_2} \right) - \phi_0 \right)(-(\delta_1-\delta_2)\psi_0+ f_1-f_2) \,.
\end{align*}
Using
$$
\| \mathcal V_{(1-\delta_1)Q+f_1}(t) - \mathcal V_{(1-\delta_2)Q+f_2}(t) \|_2 \lesssim t \left( \sup_{0\leq s\leq t} |\delta_1(s)-\delta_2(s)| + \sup_{0\leq s\leq t} \| f_1(s)-f_2(s)\|_4 \right)
$$
and \eqref{eq:contractionproof1} we obtain
\begin{align*}
& \| (\mathcal V_{(1-\delta_1)Q+f_1}-\phi_0)((1-\delta_1)\psi_0+f_1) - (\mathcal V_{(1-\delta_2)Q+f_2}-\phi_0)((1-\delta_2)\psi_0+f_2) \|_{H^{-1}} \\
& \quad \lesssim t \left( \sup_{0\leq s\leq t} |\delta_1(s)-\delta_2(s)| + \sup_{0\leq s\leq t} \| f_1(s)-f_2(s)\|_{H^1} \right).
\end{align*}
Therefore, in obvious notation,
\begin{align}
\label{eq:contractionproof3}
\sup_{0\leq t\leq\tau} |F_1^{(\delta_1)}(e_1,f_1)-F_1^{(\delta_2)}(e_2,f_2)| \lesssim \tau \left( \sup_{0\leq s\leq \tau} |\delta_1(s)-\delta_2(s)| + \sup_{0\leq s\leq \tau} \| f_1(s)-f_2(s)\|_{H^1} \right)
\end{align}
and, writing also $e_1 R f_1 - e_2 Rf_2 = (e_1-e_2)R\tfrac12(f_1+f_2) + \tfrac12(e_1+e_2) R(f_1-f_2)$,
\begin{align}
\label{eq:contractionproof4}
& \sup_{0\leq t\leq\tau} \|F_2^{(\delta_1)}(e_1,f_1)-F_2^{(\delta_2)}(e_2,f_2)\|_{H^1} \lesssim (\tau+\sigma) \sup_{0\leq t\leq \tau} \| f_1(t)-f_2(t)\|_{H^1} \notag \\
& \qquad\qquad\qquad\qquad\qquad\qquad\qquad\quad + \sigma \sup_{0\leq t\leq\tau} |e_1(t)-e_2(t)| + \tau \sup_{0\leq s\leq \tau} |\delta_1(s)-\delta_2(s)| \,.
\end{align}

We first focus on the case $\delta_1=\delta_2$. Decreasing $\sigma$ if necessary and recalling that $\tau$ is chosen small compared to $\sigma$, we see that $F$ is a contraction in $\mathcal Y\times\mathcal Z$ and therefore has a unique fixed point. This proves the first part of the claim.

It remains to prove \eqref{eq:contractionproof2}, which we prove even with the $H^1$ norm on the left side. We have, using the fixed point property and \eqref{eq:contractionproof3} and \eqref{eq:contractionproof4},
\begin{align*}
\sup_{0\leq t\leq\tau} \| f^{(\delta_1)}-f^{(\delta_2)} \|_{H^1}
& = \sup_{0\leq t\leq\tau} \| F_2^{(\delta_1)}(e^{(\delta_1)},f^{(\delta_1)}) - F_2^{(\delta_2)}(e^{(\delta_2)},f^{(\delta_2)}) \|_{H^1} \\
& \lesssim (\tau+\sigma) \sup_{0\leq t\leq\tau} \|f^{(\delta_1)}-f^{(\delta_2)} \|_{H^1} 
+ \tau \sup_{0\leq s\leq \tau} |\delta_1(s)-\delta_2(s)| \\
& \qquad\quad + \sigma \sup_{0\leq t\leq\tau} |F_1^{(\delta_1)}(e^{(\delta_1)},f^{(\delta_1)}) - F_1^{(\delta_2)}(e^{(\delta_2)},f^{(\delta_2)}) | \\
& \lesssim (\tau+\sigma) \sup_{0\leq t\leq\tau} \|f^{(\delta_1)}-f^{(\delta_2)} \|_{H^1} + \tau \sup_{0\leq s\leq \tau} |\delta_1(s)-\delta_2(s)| \,.
\end{align*}
Decreasing $\tau$ and $\sigma$ further, if necessary, we can absorb the first term into the left side and obtain \eqref{eq:contractionproof2}. This concludes the proof of Proposition \ref{reference}.\end{proof}

For $t\in[0,T_*)$ we introduce
\begin{equation}
\label{eq:chit}
\chi_t = (-\partial_x^2 +V_t-E_t)^{-1} \partial_t Q_t \,.
\end{equation}
This is well-defined since $\|Q_t\|_2=\|\psi_0\|_2$ implies $\langle Q_t,\partial_t Q_t\rangle =0$ and since $-\partial_x^2 +V_t-E_t$ is invertible as a map from the orthogonal complement of $Q_t$ to itself.

\begin{lemma}\label{boundsqchi}
Let $\psi_0,\phi_0,\dot\phi_0$ be as in Proposition \ref{reference} and let $T<T_*$. Then
$$
\| Q_t \|_\infty + \| \langle x\rangle Q_t \|_1 + \| \langle x\rangle^2 Q_t \|_2 + \| \langle x \rangle \partial_t Q_t \|_1 + \|\partial_t Q_t\|_2 + \|\partial_t^2 Q_t\|_2 \lesssim 1
$$
and
$$
\| \chi_t \|_\infty + \| \langle x\rangle \chi_t \|_1 + \| \langle x \rangle^2 \chi_t \|_2 + \| \langle x \rangle \partial_t \chi_t \|_1 + \|\partial_t \chi_t \|_2 \lesssim 1 \,.
$$
If, in addition, $\phi_0,\dot\phi_0\in \langle x\rangle^{-2} L^1(\R)$, then $V\in C^\infty([0,T_*),\langle x\rangle^{-2} L^1(\R,\R))$.
\end{lemma}

\begin{proof}
The proof of the properties of $Q_t$ and $\partial_t Q_t$ is identical to the arguments in Subsection~\ref{sec:dispprep}, since the assumption $T<T_*$ guarantees that $E_t$ stays away from zero. The bound on $\partial_t^2 Q_t$ is obtained similarly. As we discussed in Subsection~\ref{sec:dispprep}, $\partial_t Q_t$ behaves at infinity like an exponential, possibly multiplied by a linear function. Since $(-\partial_x^2 + V_t - E_t)\chi_t = \partial_t Q_t$, we can use the same ODE arguments to deduce that $\chi_t$ behaves like an exponential times a quadratic polynomial, which implies the claimed bounds for $\chi_t$. Differentiating the equation for $\chi_t$ with respect to $t$, we obtain similarly also the bounds for $\partial_t \chi_t$. The last statement about $V$ follows from \eqref{eq:solv} together with the above bounds on $Q_t$.
\end{proof}


\section{Decomposition of solutions to the Landau--Pekar equations}

After the preparations in the previous sections we now turn our attention to solutions of the Landau--Pekar equations \eqref{eq:lp}. Our goal in this section is to derive equations for $\alpha_t$ and $R_t$ appearing in the decomposition \eqref{eq:decomp} of the solution $\psi_t$.


\subsection{Decomposition of the solution and effective equations}

We assume that the initial data $(\psi_0,\phi_0,\dot\phi_0)$ are fixed as in Assumption \ref{ass:main} and we consider the corresponding solution $(V,Q)$ of \eqref{eq:eqref}, \eqref{eq:refmass} and \eqref{eq:refinitial} constructed in the previous section. Let
\begin{equation}
\label{eq:l}
L_t := -\partial_x^2 + V_t - E_t
\end{equation}
and
\begin{equation}
\label{eq:p}
P_t := 1-\|Q_t\|_2^{-2} |Q_t\rangle\langle Q_t|  =
1- \|\psi_0\|_2^{-2} |Q_t\rangle\langle Q_t| \,.
\end{equation}
Moreover, let
\begin{equation}
\label{eq:w}
W_t := - \tfrac{1}{2} \int_0^t \left( \left( |\alpha_s|^2-1\right) Q_s^2 + 2 Q_s \re(\overline{\alpha_s} R_s) + |R_s|^2 \right) \sin(t-s)\, ds \,.
\end{equation}
The following lemma describes equations for $R$ and $\alpha$ appearing in the decomposition \eqref{eq:decomp}.

\begin{lemma}[Equations for $R$ and $\alpha$]\label{effeq}
\begin{equation}
\label{eq:effeqr}
\epsilon i\partial_t R = L R + P W(\alpha Q + R)  - i\epsilon\alpha\partial_t Q - i\epsilon \|\psi_0\|_2^{-2} \langle\partial_t Q,R\rangle Q
\end{equation}
 and
 \begin{equation}
\label{eq:effeqalpha}
\partial_t \alpha = \|\psi_0\|_2^{-2} \left( \langle \partial_t Q,R\rangle - i \epsilon^{-1} \langle Q, W(\alpha Q+R)\rangle \right).
\end{equation}
\end{lemma}

\begin{proof}
Inserting decomposition \eqref{eq:decomp} into the first equation in \eqref{eq:lp} we find
\begin{align}\label{eq:effeqrproof}
\epsilon i \partial_t R & = LR + (\phi-V)(\alpha Q+R) - i\epsilon (\partial_t \alpha) Q - i\epsilon\alpha \partial_t Q \,.
\end{align}
By the second equation in \eqref{eq:lp} and the initial conditions \eqref{eq:lpinitial} we have
\begin{align*}
\phi_t= \phi_0 \cos t + \dot\phi_0 \sin t - \tfrac{1}{2} \int_0^t |\psi_s|^2 \sin(t-s) \,ds\,.
\end{align*}
Comparing this with \eqref{eq:solv} and recalling decomposition \eqref{eq:decomp}, we obtain
$$
\phi_t - V_t = - \tfrac{1}{2} \int_0^t \left( \left( |\alpha_s|^2-1\right)Q_s^2 + 2 Q_s \re(\overline{\alpha_s} R_s) + |R_s|^2 \right) \sin(t-s)\, ds = W_t \,.
$$
Inserting this into \eqref{eq:effeqrproof} yields 
\begin{align}\label{eq:effeqrproof1}
\epsilon i \partial_t R & = LR + W(\alpha Q+R) - i\epsilon (\partial_t \alpha) Q - i\epsilon\alpha \partial_t Q \,.
\end{align}

Next, we take the inner product of this equation with $Q$ and obtain, since $L$ is self-adjoint and $L Q=0$,
$$
\epsilon i \langle Q, \partial_t R\rangle = \langle Q, W(\alpha Q+R)\rangle - i\epsilon (\partial_t \alpha) \|\psi_0\|^2 - i\epsilon\alpha \langle Q,\partial_t Q\rangle \,.
$$
Using the fact that both $\|Q\|_2=\|\psi_0\|_2$ and $\langle Q,R\rangle=0$ are constant in time, we infer from the previous equation that
$$
- \epsilon i \langle \partial_t Q, R\rangle = \langle Q, W(\alpha Q+R)\rangle - i\epsilon (\partial_t \alpha) \|\psi_0\|^2 \,,
$$
which is \eqref{eq:effeqalpha}.

Finally, we apply the projection $P$ to equation \eqref{eq:effeqrproof1} and obtain, since $P$ commutes with $L$ and since $PR=R$ and $\langle Q,\partial_t Q\rangle =0$,
$$
\epsilon i P \partial_t R = LR + PW(\alpha Q+R) - i\epsilon\alpha \partial_t Q \,.
$$
Moreover, since $PR=R$ and since $\partial_t P=-\|\psi_0\|^{-2} \left( |\partial_t Q\rangle\langle Q|+ | Q\rangle\langle\partial_t Q| \right)$, we have
$$
P\partial_t R = \partial_t (PR) - (\partial_t P)R = \partial_t R + \|\psi_0\|^{-2} \langle\partial_t Q,R\rangle Q
$$
This yields \eqref{eq:effeqr}.
\end{proof}


\subsection{Extracting the leading term from $R$}

Our next goal is to remove the term $\epsilon i \alpha \partial_tQ$ from the effective equation \eqref{eq:effeqr} for $R$. Recall that the function $\chi_t$ was defined in \eqref{eq:chit}. The definition implies that
\begin{equation}
\label{eq:chiortho}
P_t \chi_t = \chi_t \,.
\end{equation}
We also note that
\begin{equation}
\label{eq:chireal}
\chi_t \ \text{is real-valued}
\end{equation}
since $Q_t$ is real-valued and since $L_t= -\partial_x^2 +V_t-E_t$, and therefore also its inverse, are reality-preserving.

 We define $\tilde R$ by
\begin{equation}
\label{eq:tilder}
\tilde R_t := R_t - i\epsilon\alpha_t \chi_t \,.
\end{equation}
It follows from \eqref{eq:decomp} and \eqref{eq:chiortho} that
\begin{equation}
\label{eq:tilderortho}
P_t \tilde R_t = \tilde R_t \,.
\end{equation}

We now derive an effective equation for $\tilde R$.

\begin{lemma}[Equation for $\tilde R$]
\begin{equation}
\label{eq:effeqrtilde}
\epsilon i \partial_t \tilde R = L \tilde R + PW(\alpha Q+i\epsilon\alpha\chi+\tilde R) - i \epsilon \|\psi_0\|^{-2} \langle\partial_t Q,\tilde R\rangle Q + \epsilon^2 (\partial_t\alpha) \chi + \epsilon^2 \alpha P \partial_t \chi \,.
\end{equation}
\end{lemma}

\begin{proof}
Inserting the definition of $\tilde R$ into equation \eqref{eq:effeqr} we obtain
\begin{equation*}
\epsilon i \partial_t \tilde R = L \tilde R + PW(\alpha Q+i\epsilon\alpha\chi+\tilde R) - i \epsilon \|\psi_0\|^{-2} \langle\partial_t Q,R\rangle Q + \epsilon^2 (\partial_t\alpha) \chi + \epsilon^2 \alpha \partial_t \chi \,.
\end{equation*}
This is the same as \eqref{eq:effeqrtilde}, since $\langle\partial_t Q,\chi\rangle + \langle Q,\partial_t\chi\rangle = \partial_t \langle Q,\chi\rangle = 0$ in view of \eqref{eq:chiortho}.
\end{proof}

We now cast the effective equation for $\tilde R$ into Duhamel form. To do so, we denote by $U(t,s)$ the $\epsilon$-adiabatic propagator for $(-\partial_x^2+V)P$, that is,
$$
i\epsilon\partial_t U(t,s) = (-\partial_x^2+ V_t) P_t U(t,s) \,,
\qquad
U(t,t) = 1 \,.
$$
Moreover, we set
$$
\tilde U(t,s) := e^{i\epsilon^{-1} \int_s^t E_{s_1}\,ds_1} U(t,s) \,,
$$
which is the propagator for $(-\partial_x^2+V)P- E$.

\begin{lemma}[Equation for $\tilde R$ in Duhamel form]
\begin{align}
\label{eq:duhamelnl}
\tilde R_t & = P_t \tilde U(t,0)\tilde R_0 + \frac{1}{i\epsilon}\int_0^t P_t \tilde U(t,s) P_s W_s \left(\alpha_s Q_s + i\epsilon\alpha_s\chi_s + \tilde R_s \right)ds \notag \\
& \qquad - \|\psi_0\|^{-2} \int_0^t P_t \tilde U(t,s) Q_s \langle\partial_sQ_s,\tilde R_s\rangle\,ds \notag \\
& \qquad - i\epsilon \int_0^t P_t \tilde U(t,s) \chi_s (\partial_s\alpha_s)\,ds - i\epsilon \int_0^t P_t\tilde U(t,s) P_s (\partial_s\chi_s)\alpha_s \,ds \,.
\end{align}
\end{lemma}

\begin{proof}
We differentiate $\tilde U(t,0)^{-1}\tilde R(t)$, use the equation for $\tilde R$ and for $\tilde U$, integrate the resulting expression and use $\tilde U(t,0) \tilde U(s,0)^{-1} = \tilde U(t,s)$ to get 
\begin{align*}
\tilde R_t & = \tilde U(t,0)\tilde R_0 + \frac{1}{i\epsilon}\int_0^t \tilde U(t,s) P_s W_s \left( \alpha_s Q_s + i\epsilon\alpha_s\chi_s + \tilde R_s \right)ds
\\
& \quad - \|\psi_0\|^{-2} \int_0^t \tilde U(t,s) Q_s \langle\partial_s Q_s,\tilde R_s\rangle\,ds \\
& \quad - i\epsilon \int_0^t \tilde U(t,s) \chi_s (\partial_s\alpha_s) \,ds - i\epsilon \int_0^t \tilde U(t,s) P_s (\partial_s \chi_s) \alpha_s \,ds \,.
\end{align*}
Finally, we apply $P_t$ to both sides, recalling \eqref{eq:tilderortho}.
\end{proof}


\section{Bounds on $\tilde R$}

Our goal in this section is to complete the proof of Theorem \ref{main}. The main step in this proof are bounds on $\tilde R$, which occupy the main part of this section.

\subsection{Control functions and their bounds}

In this short subsection we summarize the key estimates that are proved in the following subsections and that will eventually imply our main result, Theorem \ref{main}. To formulate these estimates, we introduce three control functions
\begin{align*}
\mathcal M_1(t) & := \sup_{0\leq s\leq t} \epsilon^{-1} \left\|\tilde R_s \right\|_2 \,, \\
\mathcal M_2(t) & := \sup_{0\leq s\leq t} \epsilon^{-1} \left(\max\left\{1,(\epsilon/s)^{1/2}\right\}\right)^{-1} \left\|\tilde R_s \right\|_\infty \,, \\
\mathcal M_3(t) & := \sup_{0\leq s\leq t} \epsilon^{-1} \left( \epsilon + \min\left\{ (\epsilon/s)^{1/2},(\epsilon/s)^{3/2}\right\} \right)^{-1} \left\|\langle x \rangle^{-1} \tilde R_s \right\|_\infty \,.
\end{align*}
The quantity $\mathcal M_1$ is what we are primarily interested in and our goal is to prove that $\mathcal M_1\lesssim 1$. Our strategy to proving this is to prove that, in fact, $\mathcal M_1+\mathcal M_2+\mathcal M_3\lesssim 1$. Thus, the quantities $\mathcal M_2$ and $\mathcal M_3$ appear mainly in order to close the argument. The resulting bound $\mathcal M_3\lesssim 1$ is interesting in its own right and reflects the dispersive nature of $\tilde R$, up to contributions of order $\epsilon^2$. The other bound that we obtain, namely $\mathcal M_2\lesssim 1$, is possibly non-optimal, but sufficient for our purpose.

Our bounds on these control functions read as follows.

\begin{proposition}\label{m1}
Let $T<T_*$. Then for all $t\in[0,T]$ and $\epsilon\in (0,1]$,
$$
\mathcal M_1(t)^2 \lesssim 1+ \epsilon \left( \mathcal M_3(t) + \mathcal M_1(t)^2 + \epsilon \mathcal M_1(t)\mathcal M_2(t) + \epsilon \mathcal M_1(t) \mathcal M_3(t) \right) \mathcal M_3(t) \,.
$$
\end{proposition}

\begin{proposition}\label{m2}
Let $T<T^*$. Then for all $t\in[0,T]$ and $\epsilon\in (0,1]$,
$$
\mathcal M_2(t) \lesssim 1+ \mathcal M_3(t) + \mathcal M_1(t)^2 + \epsilon \mathcal M_1(t) \mathcal M_2(t) + \epsilon \mathcal M_1(t)\mathcal M_3(t) + \epsilon \mathcal M_1(t)^2 \mathcal M_2(t) \,.
$$
\end{proposition}

\begin{proposition}\label{m3}
Let $T<T^*$. Then for all $t\in[0,T]$ and $\epsilon\in (0,1]$,
$$
\mathcal M_3(t) \lesssim 1 + \mathcal M_1(t)^2 + \epsilon \mathcal M_1(t)\mathcal M_2(t) + \epsilon \mathcal M_1(t) \mathcal M_3(t) + \epsilon^{1/2} \mathcal M_1(t)^2 \mathcal M_2(t) \,.
$$
\end{proposition}

We emphasize that the implied constants in the bounds in Propositions \ref{m1}, \ref{m2} and \ref{m3} depend on $T$.

We will prove Propositions \ref{m1}, \ref{m2} and \ref{m3} in Subsections \ref{sec:m1}, \ref{sec:m2} and \ref{sec:m3}, respectively. We will use them in Subsection \ref{sec:mainproof} to conclude the proof of Theorem \ref{main}.


\subsection{Preparations for the proof.}

In the proof of Propositions \ref{m1}, \ref{m2} and \ref{m3} we will frequently and without further mention use the bounds from Lemma \ref{boundsqchi} on $Q_t$, $\chi_t$ and their derivatives. Moreover, we will use the following bound, which follows from the usual adiabatic theorem.

\begin{lemma}
Let $T<T_*$. Then for all $0\leq s\leq t \leq T$ and $\epsilon\in(0,1]$,
\begin{equation*}
\| P_t U(t,s) Q_s \|_{H^1} \lesssim \epsilon \,.
\end{equation*}
\end{lemma}

By Sobolev's theorem, this implies, in particular, that for all $0\leq s\leq t \leq T$ with $T<T_*$,
\begin{equation}
\label{eq:adiabaticappl}
\| P_t U(t,s) Q_s \|_\infty \lesssim \epsilon \,,
\end{equation}
which will be useful later on.

\begin{proof}
Since $|\int V|f|^2\,dx |\leq \|V\|_2 \|f\|_4^2 \lesssim \|V\|_2 \|f'\|_2^{1/2} \|f\|_2^{3/2}$ and since $\|V(t)\|_2\lesssim 1$, we have for some constant $M$, independent of $t$,
$$
\|f\|_{H^1}^2 \leq 2 \langle f,(-\partial_x^2 + V_t + M) f\rangle \,.
$$
Now let $H(t) = (-\partial_x^2+V_t)P_t = (-\partial_x^2+V_t) - E_t\|\psi_0\|_2^{-2} |Q_t\rangle\langle Q_t|$. Since $|E_t|\lesssim 1$, we obtain
$$
\|f\|_{H^1}^2 \leq 2 \langle f,(H(t) + M') f\rangle \,.
$$
We obtain the bound in the lemma from Theorems \ref{adiabatic} and \ref{adiabaticenergy} applied to $f=P_tU(t,s)Q_s$. Note that we are, indeed, in the set-up of Section \ref{sec:adiabatic} with $H(t)$ defined before and with $\Phi(t)=Q_t$, $E(t)=0$, $\beta(t)=0$ (since $Q_t$ is real), $\theta(t)=0$ and $\Xi(t)=\chi_t$. The fact that the eigenvalue $0$ of $H(t)$ is simple follows from the fact that $0$ is never an eigenvalue of a one-dimensional Schr\"odinger operator with potential in $\langle x\rangle^{-1} L^1(\R)$, see \cite[Chapter 5]{Ya}. The fact that the constants in Theorems \ref{adiabatic} and \ref{adiabaticenergy} are finite follows from Lemma \ref{boundsqchi}.
\end{proof}

Finally, in the proofs of Propositions \ref{m2} and \ref{m3} in Subsections \ref{sec:m2} and \ref{sec:m3} we will apply Theorem \ref{dispersive} with the $V$ from Proposition \ref{reference}. At this point the assumption $T<T^*$, which is stronger than $T<T_*$, enters in order to satisfy the eigenvalue and non-resonance conditions in Assumption \ref{ass:disp}. It is also at this point that the assumptions $\phi_0,\dot\phi_0\in\langle x\rangle^{-2} L^1(\R)$ enter. These assumptions, together with the bounds on $Q_t$ from Lemma \ref{boundsqchi}, imply that $V$ and its derivative satisfy the properties stated in Assumption \ref{ass:disp}.


\subsection{Bounds on $W$}

We recall that $W$ was defined in \eqref{eq:w}. In this subsection we will derive bounds on $W$ in terms of our control functions. At a crucial point in our proof it will be important to use instead of $\mathcal M_3$ the modified control function $\tilde{\mathcal M}_3$ defined by
$$
\tilde{\mathcal M}_3(t) = \int_0^t \left( 1+ \epsilon^{-1} \min\left\{ \left( \frac{\epsilon}{s} \right)^{1/2}, \left( \frac{\epsilon}{s} \right)^{3/2} \right\} \right) \mathcal M_3(s)\,ds \,.
$$
Note that since
\begin{align}\label{eq:finiteint}
\int_0^t \left( 1+ \epsilon^{-1} \min\left\{ \left( \frac{\epsilon}{s} \right)^{1/2}, \left( \frac{\epsilon}{s} \right)^{3/2} \right\} \right) ds
& \leq t + \int_0^\infty \epsilon^{-1} \min\left\{ \left( \frac{\epsilon}{s} \right)^{1/2}, \left( \frac{\epsilon}{s} \right)^{3/2} \right\} ds \notag \\
& = t + 4 \,,
\end{align}
we have
\begin{equation}
\label{eq:m3refined}
\tilde{\mathcal M}_3(t) \lesssim \mathcal M_3(t) \,.
\end{equation}

\begin{lemma}\label{wbound0}
There is a real-valued function $W_0$ such that the following bounds hold on any interval $[0,T]$ with $T<T_*$,
\begin{align*}
\| \langle x \rangle (W_t-W_{0,t}) \|_2 + \| \langle x \rangle \partial_t(W_t-W_{0,t}) \|_2 & \lesssim \epsilon^2 \left( 1+ \tilde{\mathcal M}_3(t) + \mathcal M_1(t)^2 \right),\\
\| W_{0,t} \|_2 + \| \partial_t W_{0,t} \|_2 & \lesssim \epsilon^2 \mathcal M_1(t)\mathcal M_2(t), \\
\| \langle x \rangle^{-1} W_{0,t} \|_2 + \| \langle x \rangle^{-1} \partial_t W_{0,t} \|_2 & \lesssim \epsilon^3\mathcal M_1(t)\mathcal M_3(t) \,.
\end{align*}
In particular,
\begin{align*}
\|W_t\|_2 + \|\partial_t W_t\|_2 \lesssim \epsilon^2 \left( 1+ \tilde{\mathcal M}_3(t) + \mathcal M_1(t)^2 + \mathcal M_1(t)\mathcal M_2(t) \right).
\end{align*}
\end{lemma}

\begin{proof}
Inserting definition \eqref{eq:tilder} of $\tilde R$ into definition \eqref{eq:w} of $W$ we obtain the decomposition
$$
W = W_0 + W_1 + W_2 + W_3
$$
with
$$
W_{0,t} = - \tfrac12 \int_0^t |\tilde R_s|^2\sin(t-s)\,ds \,,
$$
$$
W_{1,t} = \tfrac12 \int_0^t (1-|\alpha_s|^2) Q_s^2 \sin(t-s)\,ds \,,
$$
$$
W_{2,t} = - \int_0^t Q_s \re(\overline{\alpha_s} \tilde R_s) \sin(t-s)\,ds \,,
$$
$$
W_{3,t} = -\tfrac12 \int_0^t \left( \epsilon^2 |\alpha_s|^2 \chi_s^2 + 2\epsilon \chi_s \im(\overline{\alpha_s} \tilde R_s)\right) \sin(t-s)\,ds \,.
$$
Note that for each $j$, $\partial_t W_{j,t}$ is given by the same formula as $W_j$, but with $\sin(t-s)$ replaced by $\cos(t-s)$. Consequently, $W_{j,t}+i \partial_t W_{j,t}$ is given by the same formula, but with $\sin(t-s)$ replaced by $i e^{-i(t-s)}$. In the following we shall derive bounds on the norms of $W_{j,t}+i \partial_t W_{j,t}$. Since $W_j$ is real-valued, this also implies bounds on the corresponding norms of $W_{j,t}$ and $\partial_t W_{j,t}$.

We begin with the bounds on $W_0$. We have
\begin{align*}
& \left\| W_{0,t}+i\partial_t W_{0,t} \right\|_2 \leq \tfrac{1}{2} \int_0^t \left\| \tilde R_s \right\|_2 \left\| \tilde R_s \right\|_\infty ds \\
& \qquad\qquad \leq \tfrac12 \epsilon^2 \mathcal M_1(t)\mathcal M_2(t) \int_0^t \max\left\{ 1, \left( \frac{\epsilon}{s}\right)^{1/2} \right\} ds \\
& \qquad\qquad \lesssim \epsilon^2 \mathcal M_1(t) \mathcal M_2(t) \,,\\
& \left\| \langle x \rangle^{-1} (W_{0,t}+i\partial_t W_{0,t}) \right\|_2 \leq \tfrac{1}{2} \int_0^t \left\| \tilde R_s \right\|_2 \left\| \langle x \rangle^{-1} \tilde R_s \right\|_\infty ds \\
& \qquad\qquad \leq \tfrac12 \epsilon^2 \mathcal M_1(t)\mathcal M_3(t) \int_0^t \left( \epsilon + \min\left\{ \left( \frac{\epsilon}{s}\right)^{1/2}, \left( \frac{\epsilon}{s} \right)^{3/2} \right\} \right) ds \\
& \qquad\qquad \lesssim \epsilon^3 \mathcal M_1(t) \mathcal M_3(t) \,.
\end{align*}
In the last inequality we used \eqref{eq:finiteint}.

We finally prove bounds on $\|\langle x\rangle (W_{j,t}+i\partial_t W_{j,t})\|_2$ for $j=1,2,3$. Because of the bounds on $Q_s$ and $\chi_s$ we have 
\begin{align*}
& \left\| \langle x \rangle (W_{2,t}+i\partial_t W_{2,t}) \right\|_2 \leq \int_0^t \left\|\langle x \rangle^2 Q_s \right\|_2 \left\| \langle x \rangle^{-1} \tilde R_s \right\|_\infty\,ds \\
&\qquad\qquad \lesssim \epsilon^2 \int_0^t \left( 1+ \epsilon^{-1} \min\left\{ \left( \frac{\epsilon}{s} \right)^{1/2}, \left( \frac{\epsilon}{s} \right)^{3/2} \right\} \right) \mathcal M_3(s)\,ds \\
& \qquad\qquad = \epsilon^2 \tilde{\mathcal M_3}(t) \,, \\
& \left\| \langle x \rangle (W_{3,t}+i\partial_t W_{3,t}) \right\|_2 \leq \tfrac12 \epsilon \int_0^t \left\| \langle x \rangle^2 \chi_s \right\|_2 \left( \epsilon \left\| \langle x \rangle^{-1} \chi_s \right\|_\infty + 2 \left\| \langle x \rangle^{-1} \tilde R_s \right\|_\infty\right)ds \\
& \qquad\qquad \lesssim \epsilon^2 \int_0^t \left( 1+ \left( \epsilon + \min\left\{ \left( \frac{\epsilon}{s} \right)^{1/2}, \left( \frac{\epsilon}{s}\right)^{3/2} \right\} \right) \mathcal M_3(s) \right)ds \\
& \qquad\qquad \lesssim \epsilon^2 \left( 1 + \epsilon \int_0^t \left( 1+ \epsilon^{-1} \min\left\{ \left( \frac{\epsilon}{s} \right)^{1/2}, \left( \frac{\epsilon}{s} \right)^{3/2} \right\} \right) \mathcal M_3(s)\,ds \right) \\
& \qquad\qquad = \epsilon^2 \left( 1 + \epsilon \tilde{\mathcal M}_3(t) \right).
\end{align*}
In order to bound $W_1$, we recall that the $L^2$-norm of the solution $\psi$ is constant in time. In view of the orthogonality in \eqref{eq:decomp} this implies
$$
\|\psi_0\|_2^2 = \|\psi_t\|_2^2 = |\alpha_t|^2 \|Q_t\|^2 + \|R_t\|^2 
= |\alpha_t|^2 \|\psi_0\|_2^2 + \|R_t\|_2^2  \,,
$$
that is,
$$
1-|\alpha_t|^2 = \|\psi_0\|_2^{-2} \|R_t\|_2^2 \,.
$$
This implies both $|\alpha_t|^2\leq 1$ (which follows also directly from the definition of $\alpha$ and the Schwarz inequality) and
$$
1-|\alpha_t|^2 \leq 2\|\psi_0\|_2^{-2} (\epsilon^2 |\alpha_t|^2 \|\chi_t\|_2^2 + \|\tilde R_t\|_2^2) \leq 2\|\psi_0\|_2^{-2} (\|\chi_t\|_2^2 + \mathcal M_1(t)^2) \epsilon^2 \lesssim \epsilon^2 (1+ \mathcal M_1(t)^2).
$$
Thus,
$$
\|\langle x \rangle (W_{1,t}+i\partial_t W_{1,t}) \|_2 \leq \tfrac12 \int_0^t (1-|\alpha_s|^2) \|\langle x\rangle Q_s^2\|_2 \,ds \lesssim \epsilon^2 (1+ \mathcal M_1(t)^2) \,.
$$
This concludes the proof of the lemma.
\end{proof}

\begin{corollary}\label{wbound}
There is a real-valued function $W_0$ such that the following bounds hold on any interval $[0,T]$ with $T<T_*$,
\begin{align*}
\left\| W_t \left(\alpha_t Q_t +i\epsilon\alpha_t \chi_t + \tilde R_t \right) \right\|_1 \lesssim
\epsilon^2 \!\left( 1 + \mathcal M_3(t) +\mathcal M_1(t)^2 + \epsilon \mathcal M_1(t)\mathcal M_3(t) + \epsilon \mathcal M_1(t)^2 \mathcal M_2(t) \right),
\end{align*}
\begin{align*}
& \left\| \langle x \rangle \left( W_t \left( \alpha_t Q_t +i\epsilon\alpha_t \chi_t + \tilde R_t \right) - W_{0,t} \tilde R_t \right) \right\|_1 \lesssim
\epsilon^2 \left( 1 + \tilde{\mathcal M}_3(t) + \mathcal M_1(t)^2 + \epsilon \mathcal M_1(t) \mathcal M_3(t) \right)
\end{align*}
and
$$
\left\| W_{0,t} \tilde R_t \right\|_1 \lesssim \epsilon^3 \mathcal M_1(t)^2 \mathcal M_2(t) \,.
$$
\end{corollary}

\begin{proof}
We bound
\begin{align*}
&\left\| W_t \left( \alpha_t Q_t +i\epsilon\alpha_t \chi_t + \tilde R_t \right) \right\|_1 \\
& \qquad \leq \left\| W-W_{0,t} \right\|_2 \left\| \alpha_t Q_t +i\epsilon\alpha_t \chi_t + \tilde R_t \right\|_2 \\
& \qquad\quad + \left\| \langle x \rangle^{-1} W_{0,t} \right\|_2 \left\| \langle x \rangle \left( \alpha_t Q_t +i\epsilon\alpha_t \chi_t \right) \right\|_2 + \left\| W_{0,t}\right\|_2 \left\|\tilde R_t \right\|_2 \,,
\end{align*}
\begin{align*}
&\left\| \langle x \rangle \left( W_t \left( \alpha_t Q_t +i\epsilon\alpha_t \chi_t + \tilde R_t \right) - W_{0,t} \tilde R_t \right) \right\|_1 \\
& \qquad \leq \left\| \langle x \rangle (W_t - W_{0,t}) \right\|_2 \left\| \alpha_t Q_t +i\epsilon\alpha_t \chi_t + \tilde R_t \right\|_2  + \left\| \langle x \rangle^{-1} W_{0,t} \right\|_2 \left\| \langle x \rangle^2 \left( \alpha_t Q_t+i\epsilon\alpha_t \chi_t \right) \right\|_2
\end{align*}
and
$$
\left\| W_{0,t} \tilde R_t \right\|_1 \leq \left\|W_{3,t} \right\|_2 \left\| \tilde R_t \right\|_2 \,.
$$
We now use the above bounds on the components of $W$, including the observation \eqref{eq:m3refined}, together with $|\alpha|\leq 1$ and
$$
\left\| \langle x \rangle^2 \left( \alpha_t Q_t+i\epsilon\alpha_t\chi_t \right) \right\|_2 \lesssim 1 \,,
\qquad
\left\| \alpha_t Q_t +i\epsilon\alpha_t \chi_t + \tilde R_t \right\|_2 = \left\|\psi_0\right\|_2 \lesssim 1 \,.
$$
This yields the bounds in the corollary.
\end{proof}


\subsection{Bound on $\partial_t\alpha$}

\begin{lemma}\label{alphadot}
The following bounds hold on any interval $[0,T]$ with $T<T_*$,
$$
|\partial_t \alpha_t | \lesssim \epsilon \left( 1+ \tilde{\mathcal M}_3(t) + \mathcal M_1(t)^2 + \mathcal M_1(t) \mathcal M_2(t) \right)
$$
and
\begin{align*}
\left| \partial_t (|\alpha_t|^2) \right| & \lesssim \epsilon \left( \epsilon + \min\left\{ \left( \frac{\epsilon}{t} \right)^{1/2}, \left( \frac{\epsilon}{t}\right)^{3/2} \right\} \right) \mathcal M_3(t) \\
& \quad + \epsilon^2 \left( 1+ \mathcal M_3(t) + \mathcal M_1(t)^2 + \mathcal M_1(t) \mathcal M_2(t) \right) \mathcal M_1(t) \,.
\end{align*}
\end{lemma}

While the first part of the lemma will be used in the proofs of Propositions \ref{m1}, \ref{m2} and \ref{m3} (because $\partial_t\alpha$ appears in the equation \eqref{eq:duhamelnl} for $\tilde R$), the second part of the lemma will only be used later when proving Theorem \ref{main}.

\begin{proof}
By equation \eqref{eq:effeqalpha} we have
\begin{align*}
|\partial_t\alpha| \leq \|\psi_0\|_2^{-2} \left( \|\partial_t Q\|_2 \|R\|_2 + \epsilon^{-1} \|Q\|_\infty \|W\|_2 \|\alpha Q+R\|_2 \right).
\end{align*}
The bound now follows from the bound on $\|W\|_2$ from Lemma \ref{wbound0} as well as from $\|\alpha Q+R\|_2=\|\psi_0\|_2$.

Multiplying equation \eqref{eq:effeqalpha} by $\overline\alpha$ and taking the real part, we arrive at
$$
\partial_t (|\alpha|^2) = 2\|\psi_0\|_2^{-2} \left( \langle \partial_t Q, \re(\overline\alpha \tilde R)\rangle + \epsilon^{-1} \langle Q, W \im(\overline\alpha R) \rangle \right).
$$
Note that here we used \eqref{eq:chireal}. Since $|\alpha|\leq 1$, we obtain
$$
| \partial_t (|\alpha|^2)| \leq 2\|\psi_0\|_2^{-2} \left( \| \langle x\rangle \partial_t Q\|_1 \|\langle x \rangle^{-1} \tilde R \|_\infty + \epsilon^{-1} \| Q \|_\infty \|W\|_2  \| R \|_2 \right).
$$
The bound now follows again from the bound on $\|W\|_2$ from Lemma \ref{wbound0} and \eqref{eq:m3refined}.
\end{proof}


\subsection{Bound on $\|\tilde R\|_2$}\label{sec:m1}

\begin{proof}[Proof of Proposition \ref{m1}]
We compute, using the effective equation \eqref{eq:effeqrtilde} for $\tilde R$, the self-adjointness of $L$ and the orthogonality \eqref{eq:tilderortho},
\begin{align*}
\partial_t \left\|\tilde R \right\|_2^2 & = 2 \re \langle \tilde R, \partial_t \tilde R \rangle \\
& = 2 \re \frac{1}{\epsilon i} \left\langle \tilde R,L \tilde R + PW(\alpha Q+i\epsilon\alpha\chi +\tilde R) - i\epsilon \|\psi_0\|^{-2} \langle\partial_t Q,\tilde R\rangle Q \right. \\
& \qquad\qquad\qquad \left. + \epsilon^2 (\partial_t\alpha)\chi + \epsilon^2\alpha P\partial_t\chi\right\rangle \\
& = \frac 2\epsilon \im \left\langle \tilde R,W(\alpha Q+i\epsilon\alpha\chi +\tilde R) + \epsilon^2 (\partial_t\alpha)\chi + \epsilon^2\alpha \partial_t\chi\right\rangle \\
& = \frac2\epsilon \im \left\langle \tilde R, W \left( \alpha Q+i\epsilon\alpha\chi+\tilde R \right) - W_{0} \tilde R \right\rangle + 2 \epsilon \im \left\langle \tilde R, (\partial_t\alpha) \chi + \alpha\partial_t \chi \right\rangle.
\end{align*}
In the last equality we used the fact that $W_0$ is real. We bound, using Corollary \ref{wbound} and \eqref{eq:m3refined},
\begin{align*}
& \frac 2\epsilon \im \left\langle \tilde R, W \left( \alpha Q+i\epsilon\alpha\chi+\tilde R \right) - W_0 \tilde R \right\rangle \\
& \quad \leq 2 \epsilon^{-1} \left\| \langle x \rangle^{-1} \tilde R \right\|_\infty \left\|\langle x \rangle \left( W \left( \alpha Q+i\epsilon\alpha\chi+\tilde R \right) - W_{0} \tilde R \right) \right\|_1 \\
& \quad \lesssim \epsilon^2 \left( \epsilon + \min\left\{ \left(\frac{\epsilon}{t}\right)^{1/2},\left(\frac{\epsilon}{t}\right)^{3/2} \right\} \right) \left( 1+ \mathcal M_3(t) + \mathcal M_1(t)^2 + \epsilon\mathcal M_1(t)\mathcal M_3(t) \right)\mathcal M_3(t)\, .
\end{align*}
Moreover,
$$
2\epsilon \im \left\langle \tilde R, \alpha \partial_t\chi \right\rangle \leq 2\epsilon \left\|\langle x \rangle^{-1} \tilde R \right\|_\infty \left\| \langle x\rangle \partial_t\chi\right\|_1 \lesssim \epsilon^2 \left( \epsilon + \min\left\{ \left(\frac{\epsilon}{t}\right)^{1/2},\left(\frac{\epsilon}{t}\right)^{3/2} \right\} \right) \mathcal M_3(t)
$$
and, by Lemma \ref{alphadot},
\begin{align*}
& 2\epsilon \im \langle \tilde R, \chi \partial_t \alpha \rangle \leq  2\epsilon |\partial_t\alpha| \|\langle x \rangle^{-1} \tilde R\|_\infty \|\langle x \rangle \chi\|_1  \\
& \qquad \lesssim \epsilon^3 \left( \epsilon + \min\left\{ \left(\frac{\epsilon}{t}\right)^{1/2},\left(\frac{\epsilon}{t}\right)^{3/2} \right\} \right)
\left( 1+ \mathcal M_3(t) + \mathcal M_1(t)^2 + \mathcal M_1(t)\mathcal M_2(t) \right)\mathcal M_3(t) \,.
\end{align*}
Thus,
\begin{align*}
\partial_t \left\|\tilde R \right\|_2^2 & \lesssim \epsilon^2 \left( \epsilon + \min\left\{ \left(\frac{\epsilon}{t}\right)^{1/2},\left(\frac{\epsilon}{t}\right)^{3/2} \right\} \right) \\
& \quad \times \left( 1+ \mathcal M_3(t) +\mathcal M_1(t)^2 + \epsilon \mathcal M_1(t) \mathcal M_2(t) + \epsilon\mathcal M_1(t)\mathcal M_3(t) \right) \mathcal M_3(t) \,.
\end{align*}
We integrate this bound and use \eqref{eq:finiteint} and the fact that $\|\tilde R_0\|_2 =\epsilon \|\chi_0\|_2 \lesssim \epsilon$ to conclude that
$$
\left\|\tilde R \right\|_2^2 \lesssim \epsilon^2 \left( 1+ \epsilon \left( 1+ \mathcal M_3(t) +\mathcal M_1(t)^2 +\epsilon\mathcal M_1(t)\mathcal M_2(t) + \epsilon \mathcal M_1(t) \mathcal M_3(t) \right) \mathcal M_3(t)  \right),
$$
that is,
$$
\mathcal M_1(t)^2 \lesssim 1+ \epsilon \left( 1+ \mathcal M_3(t) +\mathcal M_1(t)^2 +\epsilon\mathcal M_1(t)\mathcal M_2(t) + \epsilon \mathcal M_1(t) \mathcal M_3(t) \right) \mathcal M_3(t) \,.
$$
This implies the bound stated in the proposition.
\end{proof}


\subsection{Bound on $\|\tilde R(t)\|_\infty$}\label{sec:m2}

\begin{proof}[Proof of Proposition \ref{m2}]
We use the Duhamel formula \eqref{eq:duhamelnl} and the dispersive estimate from Theorem \ref{dispersive}, recalling \eqref{eq:chiortho}, to obtain
\begin{align*}
\left\| \tilde R_t \right\|_\infty & \lesssim \max\left\{1, \left( \frac{\epsilon}{t} \right)^{1/2} \right\}   \left\| \tilde R_0 \right\|_1 \\
&  + \frac{1}{\epsilon} \int_0^t \max\left\{1, \left( \frac{\epsilon}{t-s}\right)^{1/2} \right\} \left\| P_s W_s \left( \alpha_s Q_s+i\epsilon\alpha_s\chi_s + \tilde R_s \right) \right\|_1 \,ds \\
&  + \|\psi_0\|_2^{-2} \int_0^t \left\| P_t U(t,s) Q_s \right\|_\infty \left| \langle\partial_s Q_s,\tilde R_s\rangle\right| ds \\
&  + \epsilon \int_0^t \max\left\{1, \left( \frac{\epsilon}{t-s}\right)^{1/2} \right\} \left\| \chi_s \right\|_1 |\partial_s\alpha_s| \,ds \\
&  + \epsilon \int_0^t \max\left\{1, \left( \frac{\epsilon}{t-s}\right)^{1/2} \right\} \left\| P_s\partial_s\chi_s \right\|_1 |\alpha_s| \,ds \,.
\end{align*}
For the first term we simply use $\|\tilde R_0\|_1 = \epsilon \|\chi_0\|_1 \lesssim \epsilon$. For the second one we observe that
\begin{equation}
\label{eq:projectionbounded}
\|P_s f\|_1 \leq \left( 1 + \|\psi_0\|_2^{-2} \|Q_s\|_1 \|Q_s\|_\infty \right)\|f\|_1 \lesssim \|f\|_1 \,.
\end{equation}
With the bound from Corollary \ref{wbound} we obtain
\begin{align*}
& \frac{1}{\epsilon} \int_0^t \max\left\{1, \left( \frac{\epsilon}{t-s}\right)^{1/2} \right\} \left\| P_s W_s \left( \alpha_s Q_s+i\epsilon\alpha_s\chi_s + \tilde R_s \right) \right\|_1 \,ds \\
& \lesssim \epsilon \int_0^t \max\left\{ 1, \left( \frac{\epsilon}{t-s}\right)^{1/2} \right\} \left( 1 + \mathcal M_3(s) +\mathcal M_1(s)^2 + \epsilon \mathcal M_1(s)\mathcal M_3(s)  + \epsilon \mathcal M_1(s)^2 \mathcal M_2(s) \right) ds\\ 
& \lesssim \epsilon \left( 1 + \mathcal M_3(t) + \mathcal M_1(t)^2 + \epsilon \mathcal M_1(t)\mathcal M_3(t) + \epsilon \mathcal M_1(t)^2 \mathcal M_2(t) \right).
\end{align*}
For the third term we use \eqref{eq:adiabaticappl}, together with
\begin{equation*}
\left| \langle\partial_s Q_s,\tilde R_s\rangle\right| \leq \left\| \langle x \rangle \partial_s Q_s \right\|_1 \left\| \langle x \rangle^{-1} \tilde R_s \right\|_\infty \lesssim \epsilon \left( \epsilon + \min\left\{ \left( \frac{\epsilon}{s}\right)^{1/2}, \left( \frac{\epsilon}{s}\right)^{3/2} \right\} \right) \mathcal M_3(s) \,.
\end{equation*}
Thus, as in \eqref{eq:finiteint},
\begin{align*}
& \|\psi_0\|_2^{-2} \int_0^t \left\| P_t U(t,s) Q_s \right\|_\infty \left| \langle\partial_s Q_s,\tilde R_s\rangle\right| ds \\
& \lesssim \epsilon^2 \mathcal M_3(t) \int_0^t \left( \epsilon + \min\left\{ \left( \frac{\epsilon}{s}\right)^{1/2}, \left( \frac{\epsilon}{s}\right)^{3/2} \right\} \right)ds \lesssim \epsilon^3 \mathcal M_3(t) \,.
\end{align*}
To bound the fourth term we insert the bound from Lemma \ref{alphadot}, recall \eqref{eq:m3refined} and obtain
$$
\epsilon \int_0^t \max\left\{1, \left( \frac{\epsilon}{t-s}\right)^{1/2} \right\} \left\| \chi_s \right\|_1 |\partial_s\alpha_s| \,ds \lesssim \epsilon^2 \left( 1 + \mathcal M_3(t) + \mathcal M_1(t)^2 + \mathcal M_1(t) \mathcal M_2(t) \right)
$$
and for the fifth term we recall \eqref{eq:projectionbounded} and obtain immediately
$$
\epsilon \int_0^t \max\left\{1, \left( \frac{\epsilon}{t-s}\right)^{1/2} \right\} \left\| P_s\partial_s\chi_s \right\|_1 |\alpha_s| \,ds \lesssim \epsilon \,.
$$

To summarize, we have shown that
\begin{align*}
\left\| \tilde R_t \right\|_\infty & \lesssim \epsilon \max\left\{1, \left( \frac{\epsilon}{t}\right)^{1/2} \right\} \\
& \quad +\epsilon \left( 1+ \mathcal M_3(t) + \mathcal M_1(t)^2 + \epsilon \mathcal M_1(t) \mathcal M_2(t) + \epsilon \mathcal M_1(t)\mathcal M_3(t) + \epsilon \mathcal M_1(t)^2 \mathcal M_2(t) \right)
\end{align*}
and therefore
\begin{align*}
\mathcal M_2(t) \leq 1+ \mathcal M_3(t) + \mathcal M_1(t)^2 + \epsilon \mathcal M_1(t) \mathcal M_2(t) + \epsilon \mathcal M_1(t)\mathcal M_3(t) + \epsilon \mathcal M_1(t)^2 \mathcal M_2(t) \,,
\end{align*}
as claimed.
\end{proof}


\subsection{Bound on $\|\langle x\rangle^{-1} \tilde R(t)\|_\infty$}\label{sec:m3}

\begin{proof}[Proof of Proposition \ref{m3}]
We use the Duhamel formula \eqref{eq:duhamelnl} and the dispersive estimates from Theorem \ref{dispersive}, recalling \eqref{eq:chiortho}, to obtain
\begin{align*}
& \left\| \langle x \rangle^{-1} \tilde R_t \right\|_\infty \lesssim \min\left\{ \left( \frac{\epsilon}{t} \right)^{1/2}, \left( \frac{\epsilon}{t}\right)^{3/2} \right\}   \left\|\langle x \rangle \tilde R_0 \right\|_1 \\
& \qquad + \frac{1}{\epsilon} \int_0^t \min\left\{ \left( \frac{\epsilon}{t-s}\right)^{1/2}, \left( \frac{\epsilon}{t-s}\right)^{3/2} \right\} \\
& \qquad\qquad\qquad \times \left\| \langle x \rangle P_s \left( W_s \left( \alpha_s Q_s + i\epsilon\alpha_s \chi_s + \tilde R(s) \right) - W_{0,s} \tilde R_s \right) \right\|_1 \,ds \\
& \qquad + \frac{1}{\epsilon} \int_0^t \left( \frac{\epsilon}{t-s}\right)^{1/2} \left\| P_s W_{0,s} \tilde R_s \right\|_1 \,ds \\
& \qquad + \|\psi_0\|_2^{-2} \int_0^t \left\| \langle x \rangle^{-1} P_t U(t,s) Q_s \right\|_\infty \left| \langle\partial_s Q_s,\tilde R_s\rangle\right| ds \\
& \qquad + \epsilon \int_0^t \min\left\{ \left( \frac{\epsilon}{t-s}\right)^{1/2}, \left( \frac{\epsilon}{t-s}\right)^{3/2} \right\} \left\| \langle x \rangle \chi_s \right\|_1 |\partial_s\alpha_s | \,ds \\
& \qquad + \epsilon \int_0^t \min\left\{ \left( \frac{\epsilon}{t-s}\right)^{1/2}, \left( \frac{\epsilon}{t-s}\right)^{3/2} \right\} \left\| \langle x \rangle P_s \partial_s\chi_s \right\|_1 |\alpha_s| \,ds \,.
\end{align*}
For the first term we simply use $\left\|\langle x \rangle\tilde R_0\right\|_1 = \epsilon \left\|\langle x\rangle\chi_0 \right\|_1 \lesssim \epsilon$. For the second one we observe that
\begin{equation}
\label{eq:projectionbounded2}
\| \langle x\rangle P_s f\|_1 \leq \left( 1+ \|\psi_0\|^{-2} \|\langle x\rangle Q_s\|_1 \|\langle x\rangle^{-1} Q_s \|_\infty \right) \|\langle x\rangle f\|_1 \,.
\end{equation}
Using the bound from Lemma \ref{wbound} we obtain
\begin{align*}
& \frac{1}{\epsilon} \int_0^t \min\left\{ \left( \frac{\epsilon}{t-s}\right)^{1/2}, \left( \frac{\epsilon}{t-s}\right)^{3/2} \right\} \left\| \langle x \rangle P_s \left( W_s \left( \alpha_s Q_s + i\epsilon\alpha_s \chi_s + \tilde R(s) \right) - W_{0,s} \tilde R_s \right) \right\|_1 \,ds \\
& \lesssim \epsilon \int_0^t \min\left\{ \left( \frac{\epsilon}{t-s}\right)^{1/2}, \left( \frac{\epsilon}{t-s}\right)^{3/2} \right\} \left( 1 + \tilde{\mathcal M}_3(s) + \mathcal M_1(s)^2 + \epsilon \mathcal M_1(s) \mathcal M_3(s) \right) ds \\
& \leq \epsilon \left( 1 + \tilde{\mathcal M}_3(t) + \mathcal M_1(t)^2 + \epsilon \mathcal M_1(t) \mathcal M_3(t) \right) \int_0^t \min\left\{ \left( \frac{\epsilon}{t-s}\right)^{1/2}, \left( \frac{\epsilon}{t-s}\right)^{3/2} \right\} ds \\
& \lesssim \epsilon^2 \left( 1 + \tilde{\mathcal M}_3(t) + \mathcal M_1(t)^2 + \epsilon \mathcal M_1(t) \mathcal M_3(t) \right).
\end{align*}
The last inequality follows by a computation as in \eqref{eq:finiteint}. We emphasize that this bound is the reason why we introduced $\tilde{\mathcal M}_3$ in addition to $\mathcal M_3$. It would not be clear how to close the argument if on the right side of the previous inequality we had $\mathcal M_3$ instead of $\tilde{\mathcal M}_3$.

We bound the third term using Corollary \ref{wbound} and \eqref{eq:projectionbounded} and obtain
$$
\frac{1}{\epsilon} \int_0^t \left( \frac{\epsilon}{t-s}\right)^{1/2} \left\| P_s W_{0,s} \tilde R_s \right\|_1 \,ds \lesssim \epsilon^2 \mathcal M_1(t)^2 \mathcal M_2(t) \int_0^t \left( \frac{\epsilon}{t-s}\right)^{1/2} ds \lesssim \epsilon^{5/2} \mathcal M_1(t)^2 \mathcal M_2(t) \,.
$$

For the fourth term we argue similarly as in the previous subsection, but slightly more carefully, namely,
\begin{align*}
& \|\psi_0\|_2^{-2} \int_0^t \left\| \langle x\rangle^{-1} P_t U(t,s) Q_s \right\|_\infty \left| \langle\partial_s Q_s,\tilde R_s\rangle\right| ds \\
& \leq \|\psi_0\|_2^{-2} \int_0^t \left\| P_t U(t,s) Q_s \right\|_\infty \left| \langle\partial_s Q_s,\tilde R_s\rangle\right| ds \\
& \lesssim \epsilon^3 \int_0^t \left( 1 + \epsilon^{-1} \min\left\{ \left( \frac{\epsilon}{s}\right)^{1/2}, \left(\frac{\epsilon}{s}\right)^{3/2} \right\} \right) \mathcal M_3(s)\,ds \\
& = \epsilon^3 \tilde{\mathcal M}_3(t) \,.
\end{align*}

To bound the fifth term we insert the bound from Lemma \ref{alphadot} and obtain, arguing again as in \eqref{eq:finiteint},
\begin{align*}
& \epsilon \int_0^t \min\left\{ \left( \frac{\epsilon}{t-s}\right)^{1/2}, \left( \frac{\epsilon}{t-s}\right)^{3/2} \right\} \left\| \langle x \rangle \chi_s \right\|_1 |\partial_s\alpha_s | \,ds \\
& \lesssim \epsilon^2 \int_0^t \min\left\{ \left( \frac{\epsilon}{t-s}\right)^{1/2}, \left( \frac{\epsilon}{t-s}\right)^{3/2} \right\} \left( 1 + \tilde{\mathcal M}_3(s) + \mathcal M_1(s)^2 + \mathcal M_1(s) \mathcal M_2(s) \right) ds \\
& \lesssim \epsilon^3 \left( 1 + \tilde{\mathcal M}_3(t) + \mathcal M_1(t)^2 + \mathcal M_1(t) \mathcal M_2(t) \right),
\end{align*}
and for the sixth term we recall \eqref{eq:projectionbounded2} and obtain immediately
$$
\epsilon \int_0^t \min\left\{ \left( \frac{\epsilon}{t-s}\right)^{1/2}, \left( \frac{\epsilon}{t-s}\right)^{3/2} \right\} \left\| \langle x \rangle P_s \partial_s\chi_s \right\|_1 |\alpha_s| \,ds \lesssim \epsilon^2 \,.
$$

To summarize, we have shown that
\begin{align*}
& \left\| \langle x \rangle^{-1} \tilde R_t \right\|_\infty \lesssim \epsilon \min\left\{ \left( \frac{\epsilon}{t}\right)^{1/2}, \left( \frac{\epsilon}{t}\right)^{3/2} \right\} \\
& +\epsilon^2 \left( 1+ \tilde{\mathcal M}_3(t) + \mathcal M_1(t)^2 + \epsilon \mathcal M_1(t)\mathcal M_2(t) + \epsilon \mathcal M_1(t) \mathcal M_3(t) + \epsilon^{1/2} \mathcal M_1(t)^2 \mathcal M_2(t) \right)
\end{align*}
and therefore
\begin{align*}
\mathcal M_3(t) \leq \alpha(t) + \int_0^t \beta(s)\mathcal M_3(s)\,ds
\end{align*}
with
$$
\alpha(t) \lesssim 1+ \mathcal M_1(t)^2 + \epsilon \mathcal M_1(t)\mathcal M_2(t) + \epsilon \mathcal M_1(t) \mathcal M_3(t) + \epsilon^{1/2} \mathcal M_1(t)^2 \mathcal M_2(t)
$$
and
$$
\beta(t) \lesssim 1+ \frac{1}{\epsilon} \min\left\{ \left( \frac{\epsilon}{s}\right)^{1/2}, \left( \frac{\epsilon}{s}\right)^{3/2} \right\} .
$$
By Gronwall's lemma we conclude that
$$
\mathcal M_3(t) \leq \alpha(t) e^{\int_0^t \beta(s)\,ds} \,.
$$
Since, as in \eqref{eq:finiteint},
$$
\int_0^t \beta(s)\,ds \lesssim 1 \,,
$$
we conclude that
$$
\mathcal M_3(t) \lesssim \alpha(t) \lesssim 
1+ \mathcal M_1(t)^2 + \epsilon \mathcal M_1(t)\mathcal M_2(t) + \epsilon \mathcal M_1(t) \mathcal M_3(t) + \epsilon^{1/2} \mathcal M_1(t)^2 \mathcal M_2(t) \,,
$$
as claimed.
\end{proof}


\subsection{Proof of Theorem \ref{main}}\label{sec:mainproof}

We now show how Propositions \ref{m1}, \ref{m2} and \ref{m3}, together with some results proved along the way, imply our main result.

\begin{proof}[Proof of Theorem \ref{main}]
We fix $T<T^*$ and derive bounds uniformly in $t\in[0,T]$. If we insert the bound from Proposition \ref{m3} into Proposition \ref{m2}, we find
$$
\mathcal M_2(t) \lesssim 1 + \mathcal M_1(t)^2 + \epsilon \mathcal M_1(t)\mathcal M_2(t) + \epsilon \mathcal M_1(t) \mathcal M_3(t) + \epsilon^{1/2} \mathcal M_1(t)^2 \mathcal M_2(t) \,.
$$
Thus, introducing
$$
\mathcal M(t) := \left( \mathcal M_1(t)^2 + \mathcal M_2(t)^2 + \mathcal M_3(t)^2\right)^{1/2} \,,
$$
we have shown that
$$
\mathcal M_2(t), \mathcal M_3(t) \lesssim 1 + \mathcal M_1(t)^2 + \epsilon \mathcal M(t)^2 + \epsilon^{1/2} \mathcal M(t)^3 
\lesssim 1 + \mathcal M_1(t)^2 + \epsilon^{1/2} \mathcal M(t)^3 \,.
$$
On the other hand, from Proposition \ref{m1} we obtain
$$
\mathcal M_1(t) \lesssim 1 + \epsilon^{1/2} \mathcal M(t) + \epsilon^{1/2} \mathcal M(t)^{3/2} + \epsilon \mathcal M(t)^{3/2} \lesssim 1 + \epsilon^{1/2} \mathcal M(t)^{3/2} \,.
$$
Inserting this into the above bound on $\mathcal M_2$ and $\mathcal M_3$ we obtain
$$
\mathcal M_2(t), \mathcal M_3(t) \lesssim 1 + \epsilon \mathcal M(t)^3 + \epsilon^{1/2} \mathcal M(t)^3 \lesssim 1+ \epsilon^{1/2} \mathcal M(t)^3 \,.
$$
To summarize, we have shown that
$$
\mathcal M(t) \lesssim 1 + \epsilon^{1/2} \mathcal M(t)^3 \,.
$$
Since $\mathcal M$ is continuous and $\mathcal M(0)=\mathcal M_1(0)=\| \chi_0\|_2 \lesssim 1$ we deduce that, if $\epsilon>0$ is small enough, $\mathcal M(t)\lesssim 1$.

In particular, $\mathcal M_1(t)\lesssim 1$. Since $\| i\epsilon\alpha_t\chi_t\|\lesssim \epsilon$, this implies $\|R\|_2 \lesssim \epsilon$, as claimed. The identity $\|R\|_2 = \|\psi_0\| \sqrt{1-|\alpha|^2}$ was already derived in the proof of Lemma \ref{wbound0}.

Moreover, as shown in the proof of Lemma \ref{effeq}, $\phi-V=W$. Therefore the bound on this function and its derivative follow from Lemma \ref{wbound0} together with the bound $\mathcal M\lesssim 1$.

The bound on $\partial_t\alpha$ follows from Lemma \ref{alphadot} together with the bound $\mathcal M\lesssim 1$. The same lemma also gives
$$
|\partial_t (|\alpha|^2)| \lesssim \epsilon \left(\epsilon + \min\{ (\epsilon/t)^{1/2},(\epsilon/t)^{3/2}\} \right).
$$
This is the claimed bound for $t\geq \epsilon$. For $\epsilon\leq t$, we simply estimate $|\partial_t (|\alpha|^2)|= 2 |\re(\overline\alpha\partial_t\alpha)|\leq 2|\partial_t\alpha|$ and use the above bound on $\partial_t\alpha$. This completes the proof of the theorem.
\end{proof}


\bibliographystyle{amsalpha}

\end{document}